\documentclass[11pt,print,draftcls,onecolumn,romanappendices]{ieeecolor}
\usepackage{generic}
\usepackage[noadjust]{cite}
\usepackage{graphicx}
\usepackage{nccmath}
\usepackage{amsmath,amsfonts,amssymb}
\usepackage{mathtools}
\usepackage{url}
\usepackage{color}
\usepackage{enumerate}
\usepackage{subcaption}
\usepackage[format=hang]{caption}
\usepackage[hidelinks]{hyperref}
\usepackage{textcomp}
\def\BibTeX{{\rm B\kern-.05em{\sc i\kern-.025em b}\kern-.08em
		T\kern-.1667em\lower.7ex\hbox{E}\kern-.125emX}}
\usepackage{stfloats}
\usepackage{epstopdf}
\usepackage[cal=boondoxo]{mathalfa}
\DeclareMathAlphabet{\pazocal}{OMS}{zplm}{m}{n}

\usepackage[ruled,linesnumbered]{algorithm2e}

\makeatletter
\newcommand{\removelatexerror}{\let\@latex@error\@gobble}
\makeatother

\newcounter{storealgline}
\IncMargin{.5em}

\DontPrintSemicolon
\makeatletter
\renewcommand{\Indentp}[1]{%
	\advance\leftskip by #1
	\advance\skiptext by -#1
	\advance\skiprule by #1}%
\renewcommand{\Indp}{\algocf@adjustskipindent\Indentp{\algoskipindent}}
\renewcommand{\Indm}{\algocf@adjustskipindent\Indentp{-\algoskipindent}}
\makeatother

\newtheorem{thm}{Theorem}
\newtheorem{prob}{Problem}

\newtheorem{prop}[thm]{Proposition}
\newtheorem{cor}[thm]{Corollary}
\newtheorem{lem}[thm]{Lemma}
\newtheorem{rem}{Remark}
\newtheorem{exmp}{Example}

\DeclareMathOperator{\col}{col}
\DeclareMathOperator{\cs}{colspan}
\DeclareMathOperator{\rank}{rank}
\DeclareMathOperator{\diag}{diag}

\newcommand{\W}{\mathbb{W}}
\newcommand{\T}{\mathbb{T}}    
\newcommand{\R}{\mathbb{R}}     
\newcommand{\Z}{\mathbb{Z}}             
\newcommand{\Zp}[1]{\mathbb{Z}^+_{#1}}
\newcommand{\Znn}[1]{\mathbb{Z}^{\geq0}_{#1}}
\newcommand{\LTI}[1]{\mathfrak{L}^{#1}}

\newcommand{\n}[1]{\mathtt{n}\left(#1\right)}    
\newcommand{\lag}[1]{\mathtt{L}\left(#1\right)}

\newcommand*{\END}{\hfill\mbox{\rule[0pt]{1.3ex}{1.3ex}}}
\newcommand{\proj}[2]{\pi_{#1}\left(#2\right)}

\newcommand{\revise}[1]{{\color{black} #1}}

\newcommand{\B}{\mathfrak{B}}
\newcommand{\F}{\mathcal{F}}
\newcommand{\Hk}{\mathcal{H}}

\newcommand{\bint}[1]{{|[#1]}}

\makeatletter
\providecommand{\bigsqcap}{%
	\mathop{%
		\mathpalette\@updown\bigsqcup
	}%
}
\newcommand*{\@updown}[2]{%
	\rotatebox[origin=c]{180}{$\m@th#1#2$}%
}
\makeatother
\begin{document}

\title{Distributed Data-driven Predictive Control via Dissipative Behavior Synthesis}

\author{Yitao Yan, Jie Bao and Biao~Huang, \IEEEmembership{Fellow, IEEE}
% \thanks{*This work was supported in part by the Australian Research Council under Grant DP210101978. \emph{(Corresponding author: Jie Bao)}}
\thanks{Y. Yan and J. Bao are with the School of Chemical Engineering, UNSW Sydney, NSW 2052, Australia. (e-mail:  y.yan@unsw.edu.au; j.bao@unsw.edu.au).}
\thanks{B. Huang is with the Department of Chemical and Materials Engineering, University of Alberta, 116 St. and 85 Ave., Edmonton, AB, Canada T6G 2R3. (e-mail: biao.huang@ualberta.ca)}
}
\markboth{Draft}{}%IEEE Transactions on Automatic Control
%{Shell \MakeLowercase{\textit{et al.}}: Lorem Ipsum}

\maketitle
\global\csname @topnum\endcsname 0
\global\csname @botnum\endcsname 0
\begin{abstract}
	This paper presents a distributed data-driven predictive control (DDPC) approach using the behavioral framework. It aims to design a network of controllers for an interconnected system with linear time-invariant (LTI) subsystems such that a given global (network-wide) cost function is minimized while desired control performance (e.g., network stability and disturbance rejection) is achieved using dissipativity in the quadratic difference form (QdF). By viewing dissipativity as a behavior and integrating it into the control design as a virtual dynamical system, the proposed approach carries out the entire design process in a unified framework with a set-theoretic viewpoint. \revise{This leads to an effective data-driven distributed control design, where the global design goal can be achieved by  distributed optimization based on the local QdF conditions. The approach is illustrated by an example throughout the paper.}
\end{abstract}


\IEEEpeerreviewmaketitle
\allowdisplaybreaks

\section{Introduction}
\IEEEPARstart{T}{he} ubiquitous monitoring systems implemented in the industry are collecting \revise{a} tremendous amount of operation data. These data sets contain rich information of the processes and can be an accurate way to describe their dynamical features. As a result, there has been a gradual shift on the attention of control design from model-based to data-centric methods \cite{Stanley:2018,Wang:2019b}. On the other hand, due to the improvement of energy and economic efficiency, the design of modern industrial processes is often of large scale with different units interconnected through a very complex network. However, the associated problem with this design is that the complex interactions among subsystems make analysis and control design a challenge. This is especially so when only locally measured data sets are available because they are measured \emph{under interconnection}.

The behavioral systems theory developed by Willems \cite{Willems:1991} is expected to provide an ideal framework for data-driven analysis and control. In this framework, a dynamical system is viewed as the set of trajectories admissible through it, and a dynamical model is only \emph{one way of representing the system} instead of defining it. Input and output variables are treated equally in this framework instead of input causing output. This abandonment of causality is useful in practice because in many cases it is hard to tell which variables are the cause and which are the effect, and this is even less meaningful when only data are available. Furthermore, interconnection is viewed as variables in different subsystems sharing the same trajectory rather than signal flowing from one subsystem to another \cite{Willems:1997}. This viewpoint is especially helpful in complex interconnections because the the actual direction of signal flows may be unclear. 

Among various developments in the behavior framework, perhaps the most notable one in the direction of data-driven analysis is that the columns of the Hankel matrix constructed from a measured trajectory (with persistently exciting input) for a controllable linear time-invariant (LTI) system spans the entire finite-length trajectory space \cite{Willems:2005}. This has since led to several relaxed conditions such as extension to multiple trajectories \cite{vanWaarde:2020} and relaxed excitation condition \cite{Markovsky:2020}, as well as applications such as dissipativity verification \cite{Yan:2019a,Romer:2019} and predictive control structures \cite{Coulson:2019,Wei:2020,Berberich:2020a,Coulson:2021}. These works exclusively consider stand-alone systems instead of large-scale interconnected ones. For complex dynamics, one of the most effective control structures is distributed control due to its balance between feasibility and flexibility \cite{Tippett:2013,Tippett:2014}. Model-based control methods of interconnected systems with subsystems having parametric uncertainties have been developed in \cite{Yan:2019}. In \cite{Huang:2021}, a decentralized predictive control (which is a special case of distributed control) strategy using modified algorithm in \cite{Coulson:2019} has been applied to power systems. In all of the aforementioned data-driven control approaches, \revise{disturbance rejection was not considered}. Recently, a complete set-theoretic analysis and control design procedure has been developed in \cite{Yan:2021}, in which the network of the subsystems itself is treated as a dynamical system, resulting in a flexible plug-and-play structure that allows the subsystems to admit different types of representations (e.g., model/data mixture). The control design has also been formulated in a generic way, covering a range of different specific problems. However, \cite{Yan:2021} only provides \revise{guidelines for} the verification of the existence of desired controlled behavior implementable by distributed controllers. To the best of the authors' knowledge, a receding horizon distributed data-driven predictive control (DDPC) structure for complex interconnected systems that involves disturbance attenuation is yet to be formulated.

In this paper, a new DDPC procedure for the control of interconnected systems is developed. We assume that all subsystems are LTI but their models are unknown. The data for all subsystems are collected locally \revise{in} the presence of interconnection. The proposed approach aims to design a network of controllers that minimizes a given cost function (e.g., the economic cost of the entire plant operation and/or tracking errors) while achieving required network-wide conditions (e.g., stability and disturbance attenuation), which are specified using dissipativity with respect to a desired supply rate in quadratic difference form (QdF). Dissipativity is known for its ability to capture different types of dynamic features depending on the context and the physical meaning of the supply rate \cite{Willems:1972}. In particular, dynamic quadratic supply rates are effective in capturing the features of LTI systems \cite{Willems:1998,Kojima:2005, Tippett:2013}. \revise{Unlike previous works, in which dissipativity is treated as a property of the system \cite{Tippett:2013,Tippett:2014,Yan:2019},} this paper treats dissipativity itself as a behavior (similar to \cite{Willems:2007a}) and integrates it into the controlled system as a virtual dynamical system, resulting in a DDPC structure under a unified framework. \revise{This allows for an effective data-driven distributed control design, where the global design goal can be achieved by  distributed optimization based on local QdF conditions.}

The rest of this paper is organized as follows. Section \ref{sec:preliminaries} introduces the necessary preliminaries and formulates the problem to be solved. Section \ref{sec:interconbehavior} discusses the interconnection of finite-length behaviors. Section \ref{sec:dissipativity} gives a treatment of dissipativity in the behavioral perspective. The procedure of the proposed DDPC design and implementation is outlined in Section \ref{sec:controldesgin}. A numerical example is presented in Section \ref{sec:example} and we conclude the paper in Section \ref{sec:conclusion}.

\textbf{Notation.} We adopt the conventional notations $\R$, $\R^\mathrm{n}$, $\R^{\mathrm{m}\times \mathrm{n}}$, $\Z$, $\Zp{}$ etc. Spaces with unspecified but finite dimensions are denoted as $\R^\bullet$, $\R^{\bullet\times \mathrm{n}}$, etc. The set of all non-negative real numbers is denoted by $\R^{\geq0}$. The set of all positive (respectively, non-negative) integers no greater than $T$ is denoted as $\Zp{T}$ (respectively, $\Znn{T}$). The set of all symmetric matrices with dimension $\mathrm{n}\times\mathrm{n}$ is denoted as $\mathbb{S}^\mathrm{n}$. $I_\mathrm{n}$ and $0_{\mathrm{m}\times \mathrm{n}}$ denote, respectively, an $\mathrm{n}\times \mathrm{n}$ identity matrix and an $\mathrm{m}\times \mathrm{n}$ zero matrix. The subscripts are dropped when they are clear from the context. The Moore-Penrose inverse of a matrix $A$ is denoted as $A^\dagger$ and \revise{we} define $A^\perp\coloneqq I-A^\dagger A$, $A_\perp\coloneqq I-AA^\dagger$. For two matrices $A$ and $B$, $\col(A,B)$ and $\diag(A,B)$ stack them vertically and diagonally, respectively. For a set $A$ with $N$ elements, denote $\{A_i\}_{i=1}^N=\{A_1,A_2,\ldots,A_N\}$. Lastly, $\col\{A_i\}_{i=1}^N\coloneqq\col(A_1,A_2,\ldots,A_N)$ and analogously for $\diag\{A_i\}_{i=1}^N$. \revise{A finite dimensional space with generic variable $w$ is denoted as $\W$ and its dimension as $\mathrm{w}$} (that is, the dimension of a variable $w$ is represented by the upright w). \revise{In an interconnected system with collective variable $w$, the variable of the $i$th subsystem is denoted as $w^i$, which can be partitioned into (multi-variable) components $w^i=\col(w_1^i,w_2^i,\ldots,w_n^i)$. The space and dimension of the $j$th component, $w_j^i$, are denoted as $\mathbb{W}_j^i$ and $\mathrm{w}_j^i$, respectively.}

\section{Preliminaries and Problem Formulation}\label{sec:preliminaries}
\subsection{Behavioral Systems Theory}
In \revise{the} behavioral framework, a dynamical system is a triple $\Sigma=(\T,\W,\B)$ where $\T$ is the time axis, $\W$ is the signal space and $\B\subset\W^\T$ is the behavior. In data-driven control, trajectories obtained in data sets are typically of finite length. Without loss of generality, all trajectory samples in a data set can be assumed to have the same length, say $T$ steps. Then, the data set partially represents the behavior restricted to the interval $[1,T]$ denoted by
\begin{equation}
	\B_\bint{1,T}=\left\{w|\exists w'\in\B, \hat{w}=\hat{w}'_\bint{1,T}\right\},
\end{equation}
where $\hat{w}_\bint{1,T}\coloneqq\col(w(1),w(2),\ldots,w(T))$.

The \emph{manifest} variable $w$ contains all variables of interest such as exogenous inputs and outputs (although they are not distinguished from each other \emph{a priori}). A dynamical system is called \emph{time-invariant} if $\sigma\B\subset\B$, where $\sigma$ is a shift operator, i.e., $\sigma w\revise{(k)}=w\revise{(k+1)}$. A time-invariant system has a finite memory span, i.e., a finite time span \revise{after which the future trajectory is completely determined by its past} \cite{Willems:1991}. In this paper, \revise{we assume $\mathbb{T}\subset\mathbb{Z}^{\geq 0}$}, and that all systems are  time-invariant. In this case, the memory span can be represented in terms of the \emph{lag} of the system $\lag{\B}$, i.e., the smallest integer such that $w_\bint{k,k+\lag{\B}}\in \B_\bint{k,k+\lag{\B}}$, $\forall k\in\T$ implies $w\in\B$ \cite{Maupong:2017}. Two trajectories $\hat{w}_1,\hat{w}_2\in\B_\bint{1,L}$ can be weaved together to create a longer trajectory according to the following lemma.
\begin{lem}[\cite{Markovsky:2005}]\label{lem:weaving}
	Let $\Sigma$ be a time-invariant system. Let $\hat{w}_1,\hat{w}_2\in\B_\bint{1,L}$ with $L>\lag{\B}$. If $\hat{w}_{1\bint{L-l+1,L}}=\hat{w}_{2\bint{1,l}}$, with $l\geq\lag{\B}$, then $\col\left(\hat{w}_1,\hat{w}_{2\bint{l+1,L}}\right)\in\B_\bint{1,2L-l}$.
\end{lem}

If $w$ admits a partition $w=(w_1,w_2)$, in which, for all $w_1\in\W_1^\T$, there exists $w_2\in\W_2^\T$ such that $(w_1,w_2)\in\B$ ($\W_1$ and $\W_2$ are conformable partitions of $\W$ according to that of $w$), then $w_1$ is said to be \emph{free}. Examples of free variables include references and disturbances. This condition can be represented conveniently using the projection operator $\pi$ defined by
\begin{equation}
	\begin{split}
		&\proj{w_i}{\B}=\{w_i\mid\exists \ell_j, j\in\Zp{\mathrm{w}}\setminus\{i\},\\ &\qquad\qquad\qquad \ (\ell_1,\ldots,\ell_{i-1},w_i,\ell_{i+1},\ldots,\ell_\mathrm{w})\in\B\},
	\end{split}
\end{equation}
i.e., $w_1$ is free if $\proj{w_1}{\B}=\W_1^\T$. If $w_1$ is maximally free, that is, if none of the elements in $w_2$ are free, then $\left(w_1,w_2\right)$ is an \emph{input/output partition} of $w$. The dimensions of $w_1$ and $w_2$ are called input and output cardinalities and denoted as $\mathtt{m}(\B)$ and $\mathtt{p}(\B)$, respectively. \revise{Note that they are invariants of a given behavior, i.e., their values are independent of the representation of the behavior.}

In many cases, a dynamical system requires the aid of auxiliary variables to describe its behavior. In such a case, the full system is the quadruple $\Sigma^{full}=(\T,\W,\mathbb{L},\B^{full})$ where the \emph{full behavior} $\B^{full}\subset(\W\times\mathbb{L})^\T$, in which the auxiliary variable $\ell$ is called the \emph{latent variable}. The \emph{manifest behavior} corresponding to this full behavior is then
\begin{equation}\label{eq:manifestbehavior}
	\B=\left\{w\mid \exists \ell, \ (w,\ell)\in\B^{full}\right\},
\end{equation}
or equivalently, $\B=\proj{w}{\B^{full}}$. The latent variable $\ell$ is said to have the \emph{property of state} if, for two trajectories $(w_1,\ell_1)$, $(w_2,\ell_2)\in\B^{full}$, $\ell_1\revise{(k)}=\ell_2\revise{(k)}$ implies that the concatenation of the two trajectories at step $k$ is also a trajectory in $\B^{full}$ \cite{Polderman:1998}. In such a case, the latent variable is called a \emph{state variable} and we always have that $\lag{\B^{full}}=1$. The lowest number of state variables needed to describe $\B^{full}$ is called the \emph{state cardinality} of $\B$ and denoted by $\n{\B}$.

\subsection{Linear Time-Invariant Behaviors}
A dynamical system $\Sigma=(\T,\W,\B)$ is an LTI system if, in addition to time invariance, $\W$ is a vector space (e.g., $\W=\R^\mathrm{w}$) and $\B$ is a linear subspace of $\W^\T$ \revise{that is closed in the topology of pointwise convergence} \cite{Willems:2005,Markovsky:2006}. The set of all LTI systems with dimension $\mathrm{w}$ is denoted as $\LTI{\mathrm{w}}$ and that with unspecified dimension as $\LTI{\bullet}$. An LTI system always admits a representation of the form $R(\sigma)w=0$ where $R\in\mathbb{R}^{\bullet\times\mathrm{w}}[\sigma]$ is a polynomial matrix with variable $\sigma$. This form is called the \emph{kernel} representation of $\B$ and is denoted as $\B=\revise{\ker(R(\sigma))}$ \cite{Polderman:1998}. A kernel representation is called minimal if $R(\sigma)$ is of full row rank, and the order of such a minimal representation equals the lag $\lag{\B}$. An LTI system is controllable if $\rank(R(\lambda))$ remains unchanged for all $\lambda\in\mathbb{C}$. The set of controllable LTI systems with dimension $\mathrm{w}$ is denoted as $\LTI{\mathrm{w}}_\mathrm{contr}$. 

If an LTI system is described as the projection of a full behavior with latent variable $\ell$, then it also admits a \emph{latent variable} representation $R(\sigma)w=M(\sigma)\ell$ where $M\in\mathbb{R}^{\bullet\times\mathrm{l}}[\sigma]$. Let $\Sigma\in\LTI{\mathrm{w}}$ and suppose a measured trajectory \revise{$\hat{w}=\mathrm{col}\left(w(1),w(2),\ldots,w(T)\right)\in\B_{|[1,T]}$} is available. Then it is possible to construct a Hankel matrix of order $L$ as
\revise{\begin{equation}\label{eq:Hankelmatrix}
	\Hk_L(\revise{\hat{w}})=\begin{bmatrix}
		w(1) & w(2) & \cdots & w(T-L+1)\\
		w(2) & w(3) & \cdots & w(T-L+2)\\
		\vdots & \vdots & \ddots & \vdots\\
		w(L) & w(L+1)&\cdots& w(T)
	\end{bmatrix}.
\end{equation}}
If a set of trajectories $\mathcal{W}=\left\{\hat{w}_1,\hat{w}_2,\ldots,\hat{w}_n\right\}$ is available, then a mosaic-Hankel matrix of order $L$ can be constructed as 
\begin{equation}\label{eq:mosaicHankelmatrix}
	\Hk_L(\mathcal{W})=\begin{bmatrix}
		\Hk_L(\revise{\hat{w}}_1) & \Hk_L(\revise{\hat{w}}_2) & \cdots & \Hk_L(\revise{\hat{w}}_n)
	\end{bmatrix}.
\end{equation}
A trajectory $w$ is called \emph{persistently exciting} of order $L$ if $\rank\left(\Hk_L(\revise{\hat{w}})\right)=L\mathrm{w}$ \cite{Willems:2005}, and the trajectories in the set $\mathcal{W}$ are called \emph{collectively persistently exciting} of order $L$ if $\rank\left(\Hk_L(\mathcal{W})\right)=L\mathrm{w}$ \cite{vanWaarde:2020}. Obviously, the persistent excitation of \revise{every} trajectory in $\mathcal{W}$ implies collective persistent excitation, but the latter does not require any of its trajectories to be persistently exciting on its own. Since all subsequent developments in this paper apply to Hankel matrices and mosaic-Hankel matrices in the same way, we use them interchangeably and denote both of them as \revise{$\Hk_L(\mathcal{W})$ (for a  Hankel matrix, $\mathcal{W}$ contains only one trajectory).}

\begin{lem}[Fundamental Lemma \cite{Willems:2005,vanWaarde:2020}]\label{lem:behaviorpara}
	\revise{Let $\B\in\LTI{\mathrm{w}}_\mathrm{contr}$ and let $\mathcal{W}\subset\B_{|[1,T]}$ be a set of its trajectories. If $w$ admits an input/output partition and the input trajectories are (collectively) persistently exciting of order $L+\n{\B}$}, then $\cs\left(\Hk_L(\revise{\mathcal{W}})\right)=\B_{|[1,L]}$. In other words, for all $\hat{v}\in\B_\bint{1,L}$, there exists a vector $g\in\R^\bullet$ such that $\hat{v}=\Hk_L(\revise{\mathcal{W}})g$.
\end{lem}
\begin{rem}
		An alternative condition to check whether $\cs\left(\Hk_L(\revise{\mathcal{W}})\right)$ parameterizes $\B_{|[1,L]}$ was developed in \cite{Markovsky:2020}. It states that if $L>\lag{\B}$, then $\cs\left(\Hk_L(\revise{\mathcal{W}})\right)=\B_{|[1,L]}$ if and only if $\rank\left(\Hk_L(\revise{\mathcal{W}})\right)=\mathtt{m}\left(\B\right)L+\n{\B}$. While it is a more relaxed condition than that given in Lemma \ref{lem:behaviorpara}, verifying it is rather difficult because it requires the exact values of $\lag{\B}$ and $\n{\B}$, which are generally not easy to obtain for interconnected systems (See Lemma \ref{lem:interconproperty}).
\end{rem}

Recently, there have been rapid developments on predictive control based on Lemma \ref{lem:behaviorpara}. In essence, the Hankel matrix $\Hk_L(\revise{\mathcal{W}})$ is partitioned into input and output Hankel matrices $\Hk_L(\revise{\mathcal{U}})$ and $\Hk_L(\revise{\mathcal{Y}})$ according to the input/output partition of $w$, \revise{where $\mathcal{U}$ and $\mathcal{Y}$ are defined similarly to $\mathcal{W}$ for $u$ and $y$, respectively}. Each of them is then further partitioned into ``past'' and ``future'', i.e., $\Hk_L(\revise{\mathcal{Y}})=\col(\Hk_{past}(\revise{\mathcal{Y}}),\Hk_{future}(\revise{\mathcal{Y}}))$ and similarly for $\Hk_L(\revise{\mathcal{U}})$, correspond to the measured historical trajectory and the to-be-predicted future trajectory. The controller then aims to find a $g$ such that \cite{Markovsky:2008}
\begin{equation}\label{eq:DPC}
    \begin{bmatrix}
        \hat{y}_{past}^*\\ \hat{u}_{past}^* \\ \tilde{\hat{y}}_{future}
    \end{bmatrix}=\begin{bmatrix}
        \Hk_{past}(\revise{\mathcal{Y}}) \\\Hk_{past}(\revise{\mathcal{U}}) \\ \Hk_{future}(\revise{\mathcal{Y}})
    \end{bmatrix}g,
\end{equation}
where $\hat{y}_{past}^*$ and $\hat{u}_{past}^*$ are historical input/output trajectories and $\tilde{\hat{y}}_{future}$ is the desired future output. Further developments based on this structure include the addition of regularization condition on $g$ \cite{Coulson:2019} and slack variables \cite{Berberich:2020a} to compensate for noise in data.
\subsection{Behavioral Interconnections}
As a framework based on external variables, interconnection of dynamical systems is a very natural operation in terms of behaviors. Let $\Sigma^1=(\T,\W^1_1\times\W^1_2,\B^1)$ and $\Sigma^2=(\T,\W^2_1\times\W^2_2,\B^2)$, with \revise{$\W^1_2=\W^2_1$}. Then, the (partial) interconnection of $\Sigma^1$ and $\Sigma^2$ through $w^1_2$ and $w^2_1$ has the overall behavior 
\begin{equation*}
	\B=\left\{(w^1,w^2)\mid (w^1_1,w^1_2)\in\B^1, (w^2_1,w^2_2)\in\B^2, w^1_2=w^2_1\right\}.
\end{equation*}
In other words, interconnection is (a portion of) the variables in each subsystem \emph{sharing} the same trajectories \cite{Willems:1997,Trentelman:1999a}. If $\W^1=\W^2$, then the interconnection is called \emph{full} interconnection. In such a case, the interconnected system is denoted as $\Sigma=\Sigma^1\wedge\Sigma^2$ with behavior $\B=\B^1\cap\B^2$. All partial interconnections can be augmented into full interconnections by treating variables from other subsystems that are not shared as free variables \cite{Willems:1997}. If, on the other hand, $\W^1$ and $\W^2$ are completely different signal spaces (i.e., physically isolated), they can still be seen as ``interconnected''. Such an ``interconnection'' is denoted as $\Sigma=\Sigma^1\sqcap\Sigma^2$ with behavior $\B=\B^1\times\B^2$. Interconnections can also be viewed as isolated subsystems ``plugged'' into the network. By viewing the network interconnection as a dynamical system $\Sigma^\Pi=\revise{(\mathbb{T}, \bigtimes_{i=1}^{N}\mathbb{W}^i,\B^\Pi)}$ with its \emph{own behavior} \cite{Yan:2021}, an interconnected system can be constructed systematically as
\begin{equation}\label{eq:systemintercon}
	\Sigma=\left(\bigsqcap_{i=1}^N\Sigma^i\right)\wedge\Sigma^\Pi=\left(\T,\bigtimes_{i=1}^N\W^i,\B\right),
\end{equation}
in which $\bigsqcap_{i=1}^N\Sigma^i=\Sigma^1\sqcap\Sigma^2\sqcap\cdots\sqcap\Sigma^N$ and $\bigtimes_{i=1}^N\B^i=\B^1\times\B^2\times\cdots\times\B^N$. The interconnected behavior is then
\begin{equation}\label{eq:intercon}
	\B=\left(\bigtimes_{i=1}^{N}\B^i\right)\cap\B^\Pi.
\end{equation}
\revise{In data-driven analysis and control, trajectories measured for each subsystem are always under interconnection and as such local trajectories are from the set $\proj{w^i}{\B}$. Therefore, the behavior $\B^i$ is not available to construct $\B$ through \eqref{eq:intercon}}. Identities used in this paper concerning $\proj{w^i}{\B}$ are summarized in the following lemma.
\begin{lem}[\cite{Yan:2021}]\label{lem:interconbehaviourproj}
	Let $\Sigma$ be of the form \eqref{eq:systemintercon}, then (i) $ \proj{w^i}{\B}\subset\B^i$; (ii) $ \B=\left(\bigtimes_{i=1}^{N}\proj{w^i}{\B}\right)\cap\B^\Pi$.
\end{lem}
It is easy to see that the first statement still holds for behaviors restricted to a finite number of steps, i.e., $\proj{w^i}{\B}_\bint{1,T}\subset\B^i_\bint{1,T}$. However, the second statement does not hold for finite-length behavior in general. This point will be discussed in details in Section \ref{sec:interconbehavior}. 
\subsection{Dissipativity}
\revise{Originally defined on a system with state-space representations, a system $x(k+1)=f(x(k),u(k))$, $y(k)=h(x(k),u(k))$ is dissipative with respect to a \emph{supply rate} $s(y,u)$ if there exists a \emph{storage function} $V(x)\geq0$ such that 
	\begin{equation}\label{eq:dissconventional}
		V(x(k+1))-V(x(k))\leq s(y(k),u(k))
	\end{equation}
for all trajectories in the system \cite{Willems:1972}.} The supply rate $s(y,u)$ can admit specific forms to describe various characteristics of the system. For example, the $(Q,S,R)$-type supply rate $s(y,u)=y^\top Qy+2y^\top Su+u^\top Ru$ can be used to describe the truncated $\pazocal{L}_2$ gain from $u$ to $y$ by specifying $Q=-I$, $S=0$ and $R=\gamma^2 I$. As such, input/output stability can be represented using dissipativity.
	
\revise{The state information is often unavailable in data-driven control and the storage function can be defined on the trajectories of the manifest variable \cite{Willems:1998}.} %If $\Sigma$ is an LTI system, quadratic dynamic supply and storage are effective in capturing the dynamical features of the system \cite{Willems:1998,Willems:2002,Trentelman:2002,Kojima:2005,Wang:2019}. 
In discrete-time systems, the quadratic functional on the trajectories of $w$ is called the quadratic difference form (QdF). QdFs can be defined either on the future \cite{Kojima:2005} or the history \cite{Wang:2019}. In the context of data-driven control, we adopt the latter and define the QdF as
\begin{equation}
	\begin{split}
		Q_{\Phi}(w)\revise{(k)}&=\sum_{p=0}^{K_\Phi}\sum_{q=0}^{K_\Phi}w^\top\revise{(k-p)}\Phi_{pq}w\revise{(k-q)}\\
		&=\hat{w}_\bint{k-K_\Phi,k}^\top \widetilde{\Phi}\hat{w}_\bint{k-K_\Phi,k},
	\end{split}
\end{equation}
where $K_\Phi$ is the order of the QdF and the symmetric matrix $\widetilde{\Phi}=\widetilde{\Phi}^\top=\begin{bmatrix}
	\Phi_{K_\Phi K_\Phi} & \cdots & \Phi_{K_\Phi0}\\
	\vdots & \ddots & \vdots\\
	\Phi_{0K_\Phi} & \cdots & \Phi_{00}
\end{bmatrix}$ is the coefficient matrix. \revise{Both the supply rate $s(w)$ and storage function $V(w)$ can be represented in a QdF form, as functions of trajectories \cite{Willems:1998,Willems:2002,Trentelman:2002,Kojima:2005,Wang:2019}, i.e.,  $s(w)(k)=Q_{\Phi}(w)(k)$ and $V(w)(k)=Q_\Psi(w)(k)$.} QdF-type supply rates cover a wide range of performance specifications such as frequency-weighted disturbance attenuation \cite{Tippett:2013,Tippett:2014,Yan:2019}, filtering \cite{Trentelman:2002} and finite spectrum fault diagnosis \cite{Liu:2005,Li:2020}.

While dissipativity is originally defined as a property of a dynamical system, viewing it as a behavior in its own right provides much more insights into the dynamics it represents \cite{Willems:2007a}. Detailed treatments of dissipativity and QdFs in the behavioral context are presented in Section \ref{sec:dissipativity}.
\subsection{Problem Formulation}
In this paper we consider the control of an interconnected system $\Sigma_p$ in the form of \eqref{eq:systemintercon}. \revise{The interconnected system consists of $N_p$ linear and time-invariant subsystems $\Sigma_p^i$ ($i=1,\cdots, N_p$) and a network interconnection, which represents how the manifest variables of the subsystems are connected to each other.} The behavior of the \revise{network interconnection} $\Sigma_p^\Pi$ can be represented as $\B_p^\Pi=\revise{\ker(\Pi_p(\sigma))}$ where $\Pi_p(\sigma)$ is a \revise{polynomial} matrix. \revise{The representation for $\B_{p\bint{1,L}}^\Pi$ is $\widetilde{\Pi}_p\hat{w}_{p\bint{1,L}}=0$, where $\widetilde{\Pi}_p$ is the extended coefficient matrix of $\Pi_p(\sigma)$ whose construction is readily available in the literature (e.g., \cite{Yan:2019}).} This representation can be viewed as a generalized version of the description of topology of the network (see \cite{Tippett:2014}, for example) without explicit input/output partitions. Furthermore, \revise{it allows for the inclusion of dynamics in the network interconnection \emph{itself} such as transport delays among subsystems (e.g., mass and energy flows among process units), which is common in modern chemical plants. This view also facilitates the dissipative behavior synthesis, as detailed in Section \ref{sec:dissintercon}. For a network interconnection without delays, $\Pi_p$ reduces to a constant matrix and} $\lag{\B_p^\Pi}=0$.

The layout of the controlled system is depicted in Fig. \ref{fig:fig-plantwide}, in which $N_c$ to-be-designed distributed controllers $\Sigma_c^j$ are connected through the network interconnection $\Sigma_c^\Pi$ with behavior $\B_c^\Pi=\revise{\ker(\Pi_c(\sigma))}$. The network interconnection $\Sigma_{pc}^\Pi$ connecting the networked systems $\Sigma_p$ and controllers $\Sigma_c$ can be described by the representation
\begin{equation}\label{eq:systemcontrollernetwork}
	\Pi_{\revise{ip}}\revise{(\sigma)}w_p=\Pi_{\revise{ic}}\revise{(\sigma)}w_c.
\end{equation}
Note that we do not require $N_p=N_c$ or any specific structure on $\Pi_{\revise{ip}}$ and $\Pi_{\revise{ic}}$ in the setup. In other words, the proposed structure allows the possibilities for one controller to control multiple subsystems, several controllers to control one subsystem, or several controllers to control several subsystems cooperatively. The behavior of the controlled interconnected system with these components can then be constructed as
\begin{equation}
	\B_{pc}=\proj{w_p}{\B_{pc}^{full}}, \ \B_{pc}^{full}=(\B_p\times\B_c)\cap\B_{pc}^\Pi.
\end{equation}
In such a case, we say that the controlled behavior $\B_{pc}$ is \emph{implementable} by the controllers $\{\Sigma_c^j\}_{j=1}^{N_c}$ with $\Sigma_c^\Pi$ through $\Sigma_{pc}^\Pi$ \cite{Yan:2021}.
\begin{figure}
	\centering
	\includegraphics[width=0.5\linewidth]{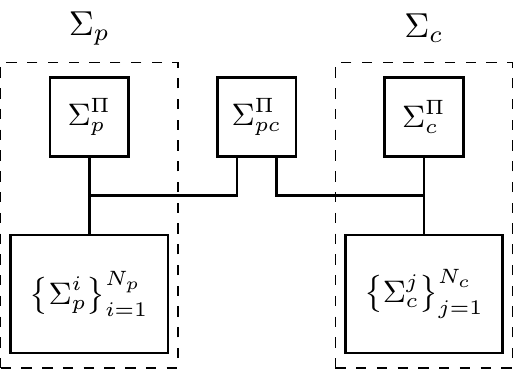}
	\caption{Structure of the Controlled System}
	\label{fig:fig-plantwide}
\end{figure}

\begin{exmp}\label{exmp:layout}
	\begin{figure}
		\centering
		\includegraphics[width=0.7\linewidth]{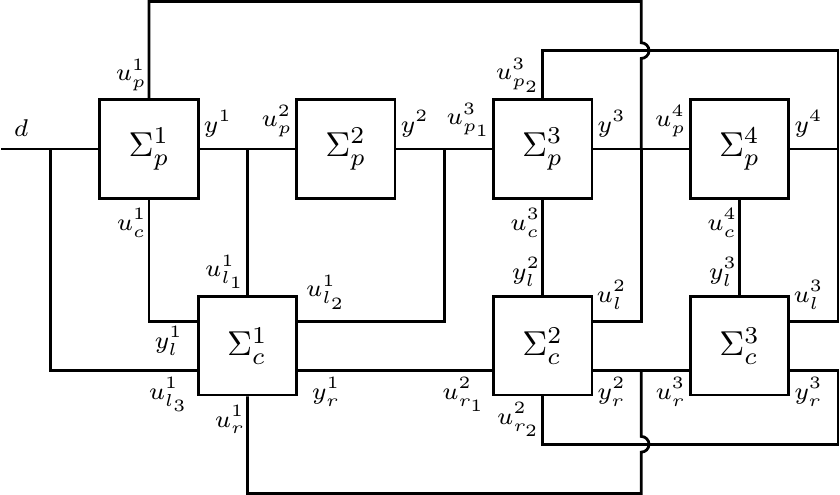}                                             
		\caption{Illustrative Diagram of an Interconnected System with Distributed Controllers}
		\label{fig:intercon}
	\end{figure}
	Consider an interconnected system depicted in Fig. \ref{fig:intercon}. It includes 4 process units (subsystems) $\{\Sigma_p^i\}_{i=1}^4$ interconnected through physical mass and energy flows which are controlled by a network of controllers $\{\Sigma_c^j\}_{j=1}^3$. For the clarity of illustration, we still adopt the conventional notations of disturbance input ($d$), interconnecting input ($u_p$), control input ($u_c$), system output ($y$), controller local input ($u_l$), remote input ($u_r$), local output ($y_l$) and remote output ($y_r$). \revise{Furthermore, we denote $u_p^3=\col(u_{p_1}^3,u_{p_3}^2)$, and similarly for $u_l^1$ and $u_r^2$ and $u_p^4$. Denote the variable of the first subsystem as $w_p^1=\col(y^1,u_p^1,u_c^1,d)$ and that of the other ones as $w_p^i=\col(y^i,u_p^i,u_c^i)$. Similarly, denote the variable of each distributed controller as $w_c^j=\col(y_l^j,y_r^j,u_l^j,u_r^j)$. Lastly, we denote $w_p=\col\{w_p^i\}_{i=1}^4$ and $w_c=\col\{w_c^j\}_{j=1}^3$.}
	
	We emphasize that, \revise{while the conventional naming of variables ($u$ and $y$) is adopted to illustrate the system and controller networks in Fig. \ref{fig:intercon}}, there is no assumed causality among the variables nor a pre-assigned ordering of subsystems. In fact, for complex interconnections, the actual direction of flow is often unclear. For instance, in a chemical process (in which context such an interconnected system is referred to as a plantwide system), the direction of physical mass/energy flows is determined by the physics in each subsystem as well as their interactions, and hence can change during operation. In this behavior-based framework, these interactions are treated as sharing of variables rather than a pre-defined set of variables of ``inputs'' and  ``outputs''. 
	
	According to Fig. \ref{fig:intercon}, we can construct the coefficient matrix $\Pi_p$, using the interconnections
	$u_p^1=y^3, \ u_p^2=y^1, \ u_p^3=\col(y^2,y^4), \ u_p^4=y^3$, as
	\begin{equation*}\label{eq:systemnetworkcoeff}
	\setlength{\arraycolsep}{3pt}
			    \Pi_p=\begin{bmatrix}
				0 & I & 0 & 0 & 0 & 0 & -I & 0 & 0 & 0 & 0 & 0 \\
				-I & 0 & 0 & 0 & 0 & I & 0 & 0 & 0 & 0 & 0 & 0 \\
				0 & 0 & 0 & 0 & \begin{pmatrix}
					-I\\0
				\end{pmatrix} & 0 & 0 & I & 0 & \begin{pmatrix}
				0\\-I
			\end{pmatrix} & 0 & 0\\
				0 & 0 & 0 & 0 & 0 & 0& -I & 0 & 0 & 0 & I & 0
			\end{bmatrix}.
			\end{equation*}
	Similarly, matrices $\Pi_{\revise{ip}}$ and $\Pi_{\revise{ic}}$ in \eqref{eq:systemcontrollernetwork} can be constructed using the system-controller interconnections
\begin{equation*}
	\begin{split}
		&u_l^1=\col(y^1, y^2, d), \ y_l^1=u_c^1,\\
	&u_l^2=y^3, \ y_l^2=u_c^3, \ u_l^3=y^4, \ y_l^3=u_c^4.
	\end{split}
\end{equation*}
Assuming that each controller shares its measured local input with its interconnected controller, except for $\Sigma_c^1$, which only shares $u_{l_2}^1$, then the network interconnection of the controllers has the behavior
\begin{equation*}
	\begin{split}
		& u_r^1=y_r^2, \ u_r^2=\col(y_r^1,y_r^3), \ u_r^3=y_r^2,\\
		& y_r^1=\begin{bmatrix}
			0 & I
		\end{bmatrix}u_l^1, \ y_r^2=u_l^2, \ y_r^3=u_l^3,
	\end{split}
\end{equation*}
through which $\Pi_c$ can be constructed similar to $\Pi_p$.
 \END
\end{exmp}

The main goal in this paper is to design the distributed controllers such that, in addition to the minimization of a global cost function, the controlled interconnected behavior is rendered dissipative with respect to a desired supply rate specified as a QdF, i.e., $s_d(w_p)=Q_{\Phi_d}(w_p)$. The problem to be solved is formulated as follows.
\revise{\begin{prob}\label{prob:problem}
    As depicted in Fig. 1, given an interconnected system $\Sigma_p$, whose subsystems $\Sigma_p^i\in\LTI{\mathrm{w}_\mathrm{p}^i}_\mathrm{contr}$, with locally measured data sets $\mathcal{W}_p =\{\mathcal{W}_p^i\}^{N_p}_{i=1}$, where $\mathcal{W}^i_p$ is the set of trajectories collected from the $i$th subsystem with manipulated variable and exogenous variable trajectories (e.g., $u_c^i$ and $d$ in Fig. \ref{fig:intercon}) persistently exciting, the controller network interconnection $\Sigma_c^\Pi$ connecting $N_c$ distributed controllers $\{\Sigma_c^j\}_{j=1}^{N_c}$, and the system-controller network interconnection $\Sigma_{pc}^\Pi$, find a set of local conditions for each distributed controller such that, when implementing these conditions under the network interconnections, the controlled interconnected system behavior
			\begin{enumerate}[1.]
				\item is dissipative with respect to a desired supply rate $s_d(w_p)$ that implies the desired global control performance (e.g., stability and disturbance rejection);
				\item admits $w_f$ as a free variable, where $w_{f}$ contains a subset of the elements in $w_p$ (e.g., $d$ in Fig. \ref{fig:intercon}).
			\end{enumerate}
			Furthermore, design a DDPC procedure that ensures the dissipativity of the interconnected system while minimizing the network cost function 
			\begin{equation}\label{eq:cost}
				J\left(\hat{w}_{c\bint{k,k+L^+}}\right)=\sum_{j=1}^{N_c}\sum_{m=0}^{L^+}{\revise{J^j}}\left(w_{c}^j(k+m)\right).
			\end{equation}
\end{prob}}

The first goal concerns the performance requirement of the control algorithm specified by a QdF supply rate. The second goal states that the chosen variables $w_f$ must remain free after control. Depending on the context, $w_f$ may contain variables such as exogenous disturbances (in disturbance attenuation or filtering) or fault signals (in fault diagnosis). The necessity of this requirement stems from the very reason why the behavioral framework is effective in handling complex interconnections: the absence of the assumption of causality. Specifically, while viewing all variables equally converts a signal flow problem into a variable sharing problem, resulting in a simple plug-and-play structure \eqref{eq:systemintercon}, such balanced view of variables does not prescribe $w_f$ as the input, so it is possible that $w_f$ becomes restricted in the controlled behavior. We therefore need the last requirement to ensure the freedom of $w_f$ after interconnecting with the controller network.
% \begin{rem}\label{rem:Naive}
% 	Since the measurements of the trajectories in the subsystems are under interconnection, these trajectories carry information of the network behavior. An intuitive approach is hence to construct local Hankel matrices $\Hk_L(\revise{\mathcal{W}}^i)$ and simply optimize $g^i$ locally such that $\hat{v}^i=\Hk_L(\revise{\mathcal{W}}^i)g^i$ gives desired local controlled behavior. This approach is incorrect because the collections of trajectories in the subsystems may not be admissible through the network, even though all subsystems have information of the network (See Theorem \ref{thm:behaviorPW}).
% \end{rem}
\section{Interconnection of Finite-Length Behaviors}\label{sec:interconbehavior}
While the relationship in Lemma \ref{lem:interconbehaviourproj}(ii) gives the method to construct the interconnected behavior from locally measured data under interconnection, the situation becomes less obvious when all data trajectories are of finite length. In particular, since the set of measured $L$-step trajectories for the interconnected system $\Sigma$ is from the set $\B_\bint{1,L}$, the locally measured trajectories are from the set $\proj{w^i}{\B_\bint{1,L}}$. Furthermore, while it is generally possible to get the exact representation for $\B^\Pi$, the fact that all subsystems are with finite-step trajectories means that the final intersection is with two finite-length behaviors. However, it can be shown (see Propositions \ref{prop:truncintersect} and \ref{prop:truncproj}) that the operation of restriction to finite-length behavior does \emph{not} commute with the intersection nor the projection operations, meaning that the combination of $L$-step local behaviors measured under interconnection and the network interconnection behavior does \emph{not} recover the $L$-step interconnected behavior in general. Fortunately, it is possible to do so for systems with time-invariant subsystems and network interconnections, provided that the trajectories are long enough. In order to show this, we firstly give upper bounds for the lag and state cardinality of the interconnected system constructed from those of its subsystems and network interconnection.
\begin{lem}\label{lem:interconproperty}
	Given an interconnected system \eqref{eq:systemintercon} with behavior \eqref{eq:intercon}, we have
	\begin{enumerate}[(i)]
		\item $\lag{\B}\leq\max\left\{\lag{\B^\Pi},\left\{\lag{\B^i}\right\}_{i=1}^N\right\}$;
		\item $\n{\B}\leq\n{\B^\Pi}+\sum_{i=1}^{N}\n{\B^i}$.
	\end{enumerate}
\end{lem}
The proof of this lemma is straightforward and therefore omitted. Note that the actual lag and state cardinality of the interconnected system $\Sigma$ depends on the subsystems and the network interconnection, so only the upper bounds of them can be computed explicitly. We are now ready to give a bound on the length of the trajectories required for the subsystems such that the way to construct the interconnected behavior, as given in Lemma \ref{lem:interconbehaviourproj}(ii), applies to finite-length behaviors as well.

\begin{thm}\label{thm:behaviorPWinvariant}
	Let $\Sigma$ be an interconnected system of the form \eqref{eq:systemintercon}, in which the subsystems are with lags $\lag{\B^i}$ and the network interconnection is with lag $\lag{\B^\Pi}$. If the number of steps $L$ satisfies $L>\max\left\{\lag{\B^\Pi},\left\{\lag{\B^i}\right\}_{i=1}^N\right\}$, then the $L$-step local behaviors \emph{under interconnection}, together with the network interconnection behavior, recover the $L$-step interconnected behavior $\B_\bint{1,L}$, i.e.,
	\begin{equation}\label{eq:interconinvariant}
		\left(\bigtimes_{i=1}^N\proj{w^i}{\B_\bint{1,L}}\right)\cap\B^\Pi_\bint{1,L}=\B_\bint{1,L}.
	\end{equation}
\end{thm}
\begin{proof}
	See Appendix \ref{appx:proofthm:behaviorPWinvariant}.
\end{proof}
\begin{rem}\label{rem:lagofintercon}
	The lower bound given in Theorem \ref{thm:behaviorPWinvariant} is not necessarily the shortest length required because the lag of the interconnected system may be smaller than the largest lag among all subsystems and the network interconnection, as stated in Lemma \ref{lem:interconproperty}(i). However, as previously discussed, the lag of the interconnected system depends on its components and their interactions, which are particularly hard to analyze in the data-driven context because all trajectories are measured under interconnections and the complete local behaviors are unavailable. The lower bound of $L$ given in this theorem, on the other hand, is a definitive bound regardless of the lag of the interconnected system.
\end{rem}

The ability to construct the interconnected behavior from the projections and the network interconnection for infinite-length behaviors has been illustrated in our previous work \cite{Yan:2021}. Here we see in Theorem \ref{thm:behaviorPWinvariant} that it is also true for finite-length behaviors if all subsystems and the network interconnection are time-invariant and the number of steps of the trajectories is longer than the largest lags among all of them. Combining Theorem \ref{thm:behaviorPWinvariant} with Lemma \ref{lem:weaving}, we immediately have the following corollary.
\begin{cor}
	Let $\Sigma$ be an interconnected system of the form \eqref{eq:systemintercon}, in which every $\Sigma^i$ is time-invariant and $\B^\Pi=\revise{\ker(\Pi(\sigma))}$. Let $\hat{w}_1^i,\hat{w}_2^i$ be two trajectories of $\B^i_\bint{1,L}$ measured under interconnection. If $\hat{w}_1,\hat{w}_2\in\B^\Pi_\bint{1,L}$ and $\hat{w}_{1\bint{L-l+1,L}}=\hat{w}_{2\bint{1,l}}$, with $l>\max\left\{\lag{\B^\Pi},\left\{\lag{\B^i}\right\}_{i=1}^N\right\}$, then
	\begin{equation}
		\hat{w}\coloneqq\col\left\{\col\left(\hat{w}_1^i,\hat{w}^i_{2\bint{l+1,L}}\right)\right\}_{i=1}^N\in\B_\bint{1,2L-l}.
	\end{equation}
\end{cor}

So far, the discussions on the construction of interconnected behavior are carried out in an entirely representation-free manner. We now direct our attention to the interconnected system concerned in this paper. Since all subsystems as well as the network interconnection are LTI systems, the interconnected system is also LTI. However, since all local trajectories are measured under interconnection, a Hankel matrix constructed from them \emph{cannot} parameterize the entire local behavior. In addition, interconnecting with other subsystems may lead to increase in the lag and state cardinality of the locally measured behavior. The following theorem gives a bound on the number of steps of persistent excitation to ensure effective parameterization and an explicit representation of the $L$-step interconnected behavior.
\begin{thm}\label{thm:behaviorPW}
	Let $\Sigma$ be an interconnected system of the form \eqref{eq:systemintercon}, in which $\Sigma^i\in\LTI{\mathrm{w}^i}_{\mathrm{contr}}$ and $\B^\Pi=\revise{\ker(\Pi(\sigma))}$. Let $\mathcal{W}^i\subset\B^i_\bint{1,T}$ be a set of trajectories in the $i$th subsystem measured under interconnection. \revise{If, in each subsystem, the trajectories for its free variable are collectively persistently exciting} of order \revise{$L+\n{\B^\Pi}+\sum_{i=1}^{N}\n{\B^i}$}, then
	\begin{equation}\label{eq:Hankelproj}
		\cs\left(\Hk_L(\revise{\mathcal{W}^i})\right)=\proj{w^i}{\B_\bint{1,L}}.
	\end{equation}
	Furthermore, if \revise{$L>\max\left\{\lag{\B^\Pi},\left\{\lag{\B^i}\right\}_{i=1}^N\right\}$}, then, for all $\hat{v}\in\B_\bint{1,L}$, there exists a vector $z\in\R^\bullet$ such that
	\begin{equation}\label{eq:Msystem}
		\hat{v}=\widehat{\Hk}_{L}(\revise{\widetilde{\Pi}}\widehat{\Hk}_{L})^\perp z\revise{\eqqcolon}\F z,
	\end{equation}
	where \revise{$\widehat{\Hk}_L=\col\left\{\diag\{\Hk_L(\mathcal{W}^i)(k)\}_{i=1}^N\right\}_{k=1}^L$, $\Hk_L(\mathcal{W}^i)(k)$ denotes the block rows of $\Hk_L(\mathcal{W}^i)$ corresponding to $v^i(k)$}. In other words, 
		\begin{equation}\label{eq:HankelPW}
		\cs(\F)=\B_\bint{1,L}.
	\end{equation}
\end{thm}
\begin{proof}
	See Appendix \ref{appx:proofthm:behaviorPW}.
\end{proof}

In Theorem \ref{thm:behaviorPW}, controllability of the subsystems has been stated as an assumption. While it has been shown in recent studies that controllability can be verified directly from data \cite{Mishra:2020}, it again requires the exact value of $\n{\B}$ which is difficult to obtain. For systems with only a known upper bound on $\n{\B}$, which is the case here, it is only possible to verify uncontrollability, but not controllability. Note that the free variables that require persistent excitation not only include $w_{fc}$ specified in Problem \ref{prob:problem} but also all manipulated variables in all subsystems, \revise{as illustrated in the following example}. 

\begin{exmp}[Example \ref{exmp:layout} continued]\label{exmp:construction}
		For the system depicted in Fig. \ref{fig:intercon}, each subsystem needs to construct Hankel matrices based on their local measurements $\mathcal{W}_p^i$ with the free variables persistently exciting. Free variables in this example includes $d$, $u_c^1$, $u_c^3$ and $u_c^4$. Note that, although $u_p^i$ is conventionally considered as an input to $\Sigma_p^i$, it is \emph{not} a free variable because it is restricted by the rest of the interconnected system. Persistent excitation is therefore not required. This again illustrates that an input to a subsystem in conventional control block diagram may not necessarily be a free variable and cannot be treated as an ``input'' in the behavioral framework. Furthermore, due to the restriction on $u_p^i$, we cannot obtain the entirety of $\B_p^i$ from $\mathcal{W}_p^i$ because, had it been a stand-alone system, $u_p^i$ should be free. However, using Theorem \ref{thm:behaviorPW}, we are able to obtain $\B_p$ from local data and the network interconnection, even though only part of the local behaviors can be obtained. \END
\end{exmp}

% \begin{rem}
% 	Similar to the requirement in Theorem \ref{thm:behaviorPWinvariant}, the bound on the order of persistent excitation provided in Theorem \ref{thm:behaviorPW} is only a sufficient condition because the actual state cardinality of the system is difficult to obtain. 
% \end{rem}

From Theorem \ref{thm:behaviorPW}, we see that the requirement \revise{on} the order of local persistent excitation \revise{involves a much larger bound} than if the subsystem were stand-alone. This is because although the set of trajectories in the projected behavior is smaller than that in the local behavior, the complexity of the former is much higher due to its convoluted ``internal dynamics'' that is the rest of the interconnected system. Furthermore, the Hankel matrix $\Hk_L(\revise{\mathcal{W}}^i)$ does \emph{not} parameterize the $L$-step local behavior $\B_\bint{1,L}^i$, but rather that of the component $w^i$ in the interconnected system $\Sigma$, which is a subset of the local behavior. According to Theorem \ref{thm:behaviorPWinvariant}, this is enough to obtain the $L$-step interconnected behavior as long as the trajectories parameterized are long enough. Note that the matrix $\F $ in \eqref{eq:Msystem} no longer maintains the structure of a Hankel matrix, but the parameterization is equally effective.

\revise{\begin{rem}
	Theorem \ref{thm:behaviorPW} highlights several points discussed in Theorem \ref{thm:behaviorPWinvariant} applied to the context of data-driven distributed control of interconnected LTI systems. Firstly, the interconnected behavior \emph{cannot} be obtained by simply stacking the Hankel matrices of the subsystems even though they represent the behavior of the subsystems under the network interconnection, i.e., $\B_\bint{1,L}\neq\cs(\widehat{\mathcal{H}}_L)$. Rather, it needs to incorporate the behavior of the network interconnection $\revise{\ker(\Pi_p(\sigma))}$. Secondly, to obtain a single representation for the common behavior of several systems (i.e., intersection of behaviors), the number of steps should be sufficiently long. This is one of the basic underlying assumptions for all developments in Section \ref{sec:controldesgin}, as noted in Section \ref{sec:preamble}.
\end{rem}}

\section{Dissipativity as Behavior}\label{sec:dissipativity}
Traditionally, dissipativity is used as an (input-output) property of a dynamical system,  describing its dynamic features \cite{Willems:1972}.  However, dissipative dynamical systems, where the \emph{behaviors} of the systems are described by dissipativity, form a class of systems in its own right. The view of dissipativity as a behavior allows it to be analyzed as a system and interconnected with other systems, leading to a clear view on the properties of a dissipative behavior that aids effective control design (see Remark \ref{rem:dissipativebehaviour} for a detailed discussion).

\subsection{Dissipative Dynamical Systems}
In the interest of data-driven analysis and control, we introduce the concept of dissipative dynamical systems in discrete time analogous to the continuous-time counterpart in \cite{Willems:2007a}. A dynamical system $\Sigma_s=(\T,\R,\B_s)$, where the manifest variable $s$ represents the rate of ``supply'' absorbed by the system, is said to be dissipative if, for any given $k_0\in\T$, there exists $C\in\R$ such that
\begin{equation}\label{eq:diss}
	-\sum_{k=k_0}^{k_1}s\revise{(k)}\leq C, \ \forall k_1\geq k_0.
\end{equation}
In this definition, the concept of dissipativity is formulated as \emph{a distinct dynamical system} with the supply rate as its manifest variable. The energy ``stored'' in the system corresponding to the supply $s$, called the storage function $V$, can serve the purpose of a latent variable. The latent variable representation of such a full behavior is given in the following proposition.

\begin{prop}\label{prop:dissipativity}
	The dynamical system $\Sigma_s=(\T,\R,\B_s)$ is dissipative if and only if there exists a latent variable dynamical system $\Sigma^{full}_s=(\T,\R,\R^{\geq0},\B_s^{full})$ with latent variable $V$ such that the full behavior $\B_s^{full}$ has the representation
	\begin{equation}\label{eq:dissstep}
		V\revise{(k)}-V\revise{(k - 1)}\leq s\revise{(k)}, \ \forall k\in\T,
	\end{equation}
	\revise{where $V\geq0$, as required by its signal space $\R^{\geq0}$.}
\end{prop}
\begin{proof}
	See Appendix \ref{appx:proofprop:dissipativity}.
\end{proof}

Despite the similarities between \eqref{eq:dissstep} and \eqref{eq:dissconventional}, the former provides several insights into the nature of dissipative dynamical systems. To begin with, the storage function $V$ has the property of state. This has been reflected in the traditional concept of dissipativity, as the original definition of dissipativity is with the storage function defined as a function of state variables \cite{Willems:1972,Willems:1972a}. Furthermore, \eqref{eq:dissstep} defines a time-invariant behavior. In other words, \emph{dissipative dynamical systems are time-invariant systems}. A direct implication of these two insights is that the lag of $\B^{full}$ defined in \eqref{eq:dissstep} is always 1. \revise{In addition, the dissipative behavior described by \eqref{eq:dissstep} is nonlinear.}

\begin{rem}\label{rem:dissipativebehaviour}
	By viewing dissipativity as the behavior of a dissipative dynamical system rather than a property of a given dynamical system, it possesses various system properties such as time-invariance and lag. As will be discussed in the next section, this view facilitates the analysis on the lag of the dissipative behavior of a system, and hence the length of the receding horizon during optimization.
\end{rem}
\subsection{Interconnection with a Dissipative Dynamical System}\label{sec:dissintercon}
While the primary application of dissipativity is to extract useful dynamical features of a given dynamical system for effective control design, the dynamical features can be understood much more clearly if we treat the feature extraction process as the interconnection of the given system $\Sigma$ and a \emph{virtual} dissipative dynamical system $\Sigma_s$ through some network interconnection $\Sigma^\Pi$. Since the lag of $\Sigma_s$ without the aid of the storage function is difficult to obtain, we treat the storage function as a manifest variable as well. In other words, the dissipative dynamical system used here is \emph{defined} by the triple $\Sigma_{sV}=(\T,\R^2,\B_{sV})$, where the manifest variable is the pair $(s,V)$ and the behavior $\B_{sV}$ is described by \eqref{eq:dissstep}. By doing so, it is clear that $\lag{\B_{sV}}=1$. The behavior of the supply rate $s$, in the sense of \eqref{eq:diss}, can be extracted from $\B_{sV}$ as $\B_s=\proj{s}{\B_{sV}}$. 

Suppose a dynamical system $\Sigma$ is given by the triple $\Sigma=(\T,\W,\B)$. The \emph{virtual network interconnection} connecting the physical system $\Sigma$ and the virtual system $\Sigma_{sV}$ is given by $\Sigma_s^\Pi=(\T,\W\times\R^2,\B_{s}^\Pi)$, where the network interconnection behavior $\B_{s}^\Pi$ is to be specified. As a result, the ``interconnected'' system $\Sigma_{ps}$ has the behavior
\begin{equation}\label{eq:invariant+diss}
	\B_{ps}=(\B\times\B_{sV})\cap\B_{s}^\Pi.
\end{equation}
The sub-behavior of $\B$ that is dissipative with respect to the supply rate $s$ with storage function $V$ specified by $\B_{s}^\Pi$, or $(s,V)$-dissipative, is given by $\proj{w}{\B_{ps}}$. Obviously, if $\B\subset\proj{w}{\B_{ps}}$, then the entire behavior $\B$ is $(s,V)$-dissipative. Since $\proj{w}{\B_{ps}}\subset\B$ by construction (see Lemma \ref{lem:projinout}), we have that $\B$ is $(s,V)$-dissipative if and only if $\B=\proj{w}{\B_{ps}}$. In the traditional sense, the network interconnection behavior $\B_{s}^\Pi$ specifies the storage function and supply rate associated with the dynamical system $\Sigma$. It will become apparent that viewing the search for dissipativity of a given system as an interconnection facilitates the analysis and control design.

As explained in the beginning of this section, $\Sigma_{s}^\Pi$ specifies the interconnection between $w$ and the pair $(s,V)$. If the network interconnection is specified using QdFs, i.e.,
\begin{equation}\label{eq:dissnetwork}
	s\revise{(k)}=Q_\Phi(w)\revise{(k)}, \ V\revise{(k)}=Q_\Psi(w)\revise{(k)},
\end{equation}
with orders $K_\Phi$ and $K_\Psi$, respectively, then the resulting behavior $\B_{s}^\Pi$ is obviously time-invariant. The upper bound of the lag can be obtained as $\lag{\B_{s}^\Pi}\leq\max\{K_\Phi,K_\Psi\}$. Suppose, in addition, that $\Sigma\in\LTI{\mathrm{w}}$ and that $\B_{\bint{1,L}}=\mathrm{colspan}\left(\F\right)$. Since $\Sigma$ is time-invariant, viewing any step within the $L$-step trajectory as ``current step'' does not change the representation. Suppose that the trajectories are re-indexed to the interval $[k-L^-,k+L^+]$, with $L^++L^-+1=L$. Then the behavior $\proj{w}{\B_{ps}}$ has the representation
\begin{subequations}\label{eq:LTI+diss}
	\begin{align}
		\hat{w}_{k}\coloneqq\hat{w}_{\bint{k-L^-,k+L^+}}&=\F z_{k},\label{eq:LTI+diss1}\\
		Q_{\Phi}(w)\revise{(k)}-Q_{\nabla\Psi}(w)\revise{(k)}&\geq0 ,\label{eq:LTI+diss2}
	\end{align}
\end{subequations}
where $Q_{\nabla\Psi}(w)\revise{(k)}\coloneqq Q_{\Psi}(w)\revise{(k)}-Q_{\Psi}(w)\revise{(k-1)}$, called the \emph{rate of change} of the storage, has the coefficient matrix
\begin{equation}
	\nabla\widetilde{\Psi}\coloneqq\diag(0_\mathrm{w},\widetilde{\Psi})-\diag(\widetilde{\Psi},0_\mathrm{w}).
\end{equation}
It is not difficult to see that bound of the lag of this behavior \eqref{eq:LTI+diss} is given by
\begin{equation}\label{eq:lagdissproj}
	\lag{\proj{w}{\B_{ps}}}\leq\max\{\lag{\B},K_\Phi,K_\Psi+1\}.
\end{equation}

We now present conditions to deduce the dissipativity of $\B$ on the infinite time axis from that on a finite interval. The key issue lies in the number of steps needed to verify the dissipativity condition.
\begin{prop}\label{prop:Hankeldiss}
	Let $\B\in\LTI{\mathrm{w}}$ with $\B_{\bint{k+L^-,k+L^+}}$ represented by \eqref{eq:LTI+diss1}. Let $\B_{sV}$ be a dissipative behavior with representation \eqref{eq:dissstep} and let the interconnection between $\B$ and $\B_{sV}$ be represented by \eqref{eq:dissnetwork}. Assume that $L^-=\max\{K_\Phi,K_\Psi+1\}$. If $L^->\lag{\B}$, $\widetilde{\Psi}\geq0$ and, for all $l\in\Znn{L^+}$,\revise{
		\begin{equation}\label{eq:PhiLdissSOS}
			z_k^\top\F_l^\top (\widehat{\Phi}-\nabla\widehat{\Psi})\F_lz_k\geq0
		\end{equation}
	holds for any $z_k$, where $\F_l$ denotes the sub-matrix of} $\F$ corresponding to the trajectory segment $\hat{w}_{\bint{k-L^-+l,k+l}}$, $\widehat{\Phi}=\diag(0_{k_\Phi\mathrm{w}},\widetilde{\Phi})$, $\nabla\widehat{\Psi}=\diag(0_{k_\Psi\mathrm{w}}, \nabla\widetilde{\Psi})$, $k_\Phi=L^--K_\Phi$ and $k_\Psi=L^--K_\Psi-1$, then $\B=\proj{w}{\B_{ps}}$, where $\B_{ps}$ is given in \eqref{eq:invariant+diss}. In other words, $\B$ is $(Q_\Phi(w),Q_\Psi(w))$-dissipative. \revise{Furthermore, if $L^+=0$, then \eqref{eq:PhiLdissLMI} reduces to}
	\begin{equation}\label{eq:PhiLdissLMI}
	    \revise{\F^\top (\widehat{\Phi}-\nabla\widehat{\Psi})\F\geq0.}
	\end{equation}
\end{prop}
\begin{proof}
	See Appendix \ref{proofprop:Hankeldiss}.
\end{proof}
\begin{rem}
	The assumption that $L^-=\max\{K_\Phi,K_\Psi+1\}$ does not lead to loss of generality because, if they were not equal, we could extend the shorter one to match the length either through weaving of the coefficient matrix (if $w_p$ were not long enough) or by starting with a negative value of $l$ (if the orders of the QdFs were not high enough).
\end{rem}

The conditions in Proposition \ref{prop:Hankeldiss} claim that if \eqref{eq:LTI+diss2} is satisfied for all $L$-step trajectories described in \eqref{eq:LTI+diss1} and the ``past trajectory'' is long enough, then the dissipativity condition given by the pair $(Q_\Phi(w),Q_\Psi(w))$ is valid for the entire behavior $\B_p$. It is worth pointing out that the requirement of the length of the ``past trajectory'' is more than needed for the dissipativity to be valid. In fact, $\B$ is $(Q_\Phi(w),Q_\Psi(w))$-dissipative as long as the second condition in Proposition \ref{prop:Hankeldiss} holds for $L>\max\{\lag{\B},K_\Phi,K_\Psi+1\}$. The requirement for longer trajectory is for the purpose of predictive control, in that sufficient overlapping between consecutive receding horizons is needed to ensure effective weaving. 
\subsection{More on QdF-based Behaviors}
In the previous section we \revise{saw} that the representation of dissipative behavior with quadratic supply and storage, i.e., behavior described by \eqref{eq:LTI+diss2}, has the general form
\begin{equation}\label{eq:QTI}
	\hat{w}^\top M\hat{w}\geq0,
\end{equation}
where $M\in\mathbb{S}^\bullet$ is a symmetric coefficient matrix. In this section, we give some further treatments to this type of systems. To begin with, it is easy to see that the behavior defined by \eqref{eq:QTI} is time-invariant and that the lag $\lag{\B}$ is given by the order of the QdF.

\revise{Let $w$ admit a partition $w=(w_1,w_2)$. The matrix $M$ can be partitioned accordingly, resulting in the representation}
\begin{equation}\label{eq:dissmanifest}
	\begin{bmatrix}
		\hat{w}_1\\ \hat{w}_2
	\end{bmatrix}^\top \begin{bmatrix}
		Q&S\\S^\top &R
	\end{bmatrix}\begin{bmatrix}
		\hat{w}_1\\ \hat{w}_2
	\end{bmatrix}\geq0,
\end{equation}
where $Q$, $R$ are symmetric matrices. The immediate question is that whether the \revise{behavior of one of the partitions, say $w_1$,} has a similar form. Interestingly, it is \emph{not} always the case, and the exact manifest behavior depends on the signature of $R$.
\begin{prop}\label{prop:dissmani}
	Let $\Sigma=(\T,\W,\B)$ be a latent variable dynamical system whose $L$-step behavior $\B_\bint{1,L}$ has the representation \eqref{eq:dissmanifest}. Then the \revise{behavior of $w_1$} is given by
	\begin{equation}\label{eq:dissmani}
		\begin{cases}
			\hat{w}_1^\top \left(Q-SR^\dagger S^\top \right)\hat{w}_1\geq0, & \text{if } R\leq0 \text{ and } R_\perp S^\top \hat{w}=0;\\
			\W_1^{{\Zp{L}}}, & \text{otherwise}.
		\end{cases}
	\end{equation}
\end{prop}
\begin{proof}
	See Appendix \ref{appx:proofprop:dissmani}.
\end{proof}

From Proposition \ref{prop:dissmani}, we see that the behavior of $w$ is only restricted when $R\leq0$, in which case the manifest behavior is quadratic. Note that this theorem only gives the manifest behavior for $L$ steps. This is because, like LTI systems, the dynamics of the latent variable are reflected by the extension of lag in the manifest behavior. As a result, the manifest behavior should have longer lag than $L$, hence it cannot uniquely define the behavior on the infinite time axis. Fortunately, we will see in subsequent sections that this is enough to carry out effective control design. From the behavior given in \eqref{eq:dissmani}, we can immediately obtain a simple and useful sufficient condition to verify that $w$ is free on the interval. This is given in the following corollary.
\begin{cor}\label{cor:dissfree}
	Let $\B_\bint{1,L}$ be described by \eqref{eq:dissmanifest}. If \revise{$Q\geq0$}, then $\hat{w}_1$ is free.
\end{cor}
\revise{Note that this corollary is not limited to partition of variables but rather any partition (e.g., partition of past/future trajectories). This determination of free trajectories in QdF-based behaviors is one of the key ingredients to the solution of Problem \ref{prob:problem} (as shown in Theorem \ref{thm:controlledexist}).}

\section{Distributed Control Design}\label{sec:controldesgin}
In this section, we develop the DDPC procedure based on the previously collected trajectory set $\mathcal{W}_p$ and the prescribed desired supply rate $s_d=Q_{\Phi_d}(w_p)$ with coefficient matrix $\widetilde{\Phi}_d$. In the context of behavior, control is viewed as an interconnection to pose extra constraints on the possible outcomes of the uncontrolled system. In other words, the desired controlled behavior has to be already contained in the original system in order for it to be possible \cite{Willems:2002,Yan:2021}. Therefore, the key issue of effective control design lies in the verification of the feasibility of the desired controlled behavior. For stand-alone systems, the control algorithm will be ready once feasibility is verified. The design of distributed control algorithms, on the other hand, is much less straightforward. In this section, we first give conditions for feasibility, then we present the structure for the proposed DDPC algorithm. 
\subsection{Preamble}\label{sec:preamble}
Throughout this section, we will use the relationship $\B_{1\bint{1,L}}\cap\B_{2\bint{1,L}}=(\B_1\cap\B_2)_\bint{1,L}$ without further explanations. While this is \emph{not} true in general, it is true when all systems considered are time-invariant and \revise{when} $L$ is large enough (see Proposition \ref{prop:truncintersect}). An implicit assumption is therefore that the trajectories \revise{are sufficiently long}.

As claimed in Theorem \ref{thm:behaviorPW}, the combination of the local restricted behaviors parameterized by Hankel matrices of collectively persistently exciting measurements and the network interconnection behavior give the $L$-step interconnected behavior. Following the construction of $\F $ in \eqref{eq:Msystem}, the interconnected behavior $\B_p$ can be constructed as
\revise{\begin{equation}\label{eq:systemManifest}
    \hat{w}_{p_k}=\F_pz_{p_k}.
\end{equation}
\revise{Combined with the behaviors of network interconnections $\B_{pc}^\Pi$ and $\B_c^\Pi$, the largest possible set of trajectories (called the largest possible behavior)} of the controlled system is
\begin{equation}\label{eq:controlledLargest}
    \begin{bmatrix}
        \widetilde{\Pi}_{ic}\\ \widetilde{\Pi}_c
    \end{bmatrix}\hat{w}_{c_k}=\begin{bmatrix}
        \widetilde{\Pi}_{ip}\F_p\\ 0
    \end{bmatrix}z_{p_k}.
\end{equation}

For the clarity of presentation, subsequent discussions in this particular section are carried out with two regularity assumptions. Firstly, we focus on the dissipativity design in the interval $[k-L^-,k]$. This means that, within each horizon, we ensure the dissipativity condition for the $k$th step to guarantee the control performance for the step that is actually implemented while planning further steps into the future are  according to a cost function. As shown in Proposition \ref{prop:Hankeldiss}, this allows us to formulate the results in terms of a single set of linear matrix inequalities (LMIs). Secondly, let the LHS coefficient matrix of \eqref{eq:controlledLargest} be invertible. This case occurs when each trajectory of $w_p$ correspond to a single (but not necessarily unique) trajectory of $w_c$. In the general case, two additional situations may arise: If the LHS coefficient matrix is of full row rank, then multiple trajectories of $\hat{w}_c$ may correspond to the same $\hat{w}_p$, giving the controllers more freedom; If the LHS coefficient matrix is rank-deficient, then the controllers may have already excluded a set of trajectories of $\B_p$ (which may include ones that satisfy the global performance requirements) from the controlled behavior before dissipativity synthesis. The general results with both assumptions relaxed are presented in Appendix \ref{appx:controlGeneral}.

\subsection{Existence of Desired Controlled Behavior}\label{sec:existence}
Since all results in this section are within the interval $[k-L^-,k]$, we omit the dependency on $k$ unless specification is required. Under the regularity assumptions, all possible trajectories of $w_c$ in \eqref{eq:controlledLargest} can be represented as
\begin{equation}\label{eq:controllerManifest}
    \hat{w}_c=\begin{bmatrix}
        \widetilde{\Pi}_{ic}\\ \widetilde{\Pi}_c
    \end{bmatrix}^{-1}\begin{bmatrix}
        \widetilde{\Pi}_{ip}\F_p\\ 0
    \end{bmatrix}z_p\eqqcolon\F_cz_p,
\end{equation}
and that of $\hat{w}_c^j$ as the column span of the corresponding rows of $\F_c$, denoted as $\F_c^j$. Note that, different from this special case, in which \eqref{eq:systemManifest} and \eqref{eq:controllerManifest} share the same variable $z_p$, variables parameterizing the space of trajectories of $\hat{w}_p$ and that of $\hat{w}_c$ are different in general.

Within each horizon, for a given $z_p$ corresponding to a past trajectory, denoted as $\hat{w}_{c-}$ (with corresponding coefficient matrix $\F_{c-}$), and a current step of the manipulated variables, denoted as $w_m(k)$ (with corresponding coefficient matrix $\F_{mk}$), all solutions of $z_p$ that lead to the same $\hat{w}_{c-}$ and $w_m(k)$ can be computed as (See Lemma \ref{lem:Ax=b})
\begin{equation}\label{eq:zpCal}
    z_p'=\begin{bmatrix}
        \F_{c-}\\ \F_{mk}
    \end{bmatrix}^\dagger\begin{bmatrix}
        \hat{w}_{c-}\\ w_m(k)
    \end{bmatrix} + \begin{bmatrix}
        \F_{c-}\\ \F_{mk}
    \end{bmatrix}^\perp z_h,
\end{equation}
where $z_h$ is an arbitrary vector. In other words, all possible trajectories corresponding to the implementation of $w_m(k)$ in the trajectory
\begin{equation}\label{eq:controlImp}
    \hat{w}_c=\F_c\begin{bmatrix}
        \F_{c-}\\ \F_{mk}
    \end{bmatrix}^\dagger\begin{bmatrix}
        \F_{c-}\\ \F_{mk}
    \end{bmatrix}z_p\eqqcolon\F_{cc}z_p
\end{equation}
can be represented as
\begin{equation}\label{eq:wpControlled}
    \begin{split}
        \hat{w}_p&=\F_p\begin{bmatrix}
        \F_{c-}\\ \F_{mk}
    \end{bmatrix}^\dagger\begin{bmatrix}
        \F_{c-}\\ \F_{mk}
    \end{bmatrix}z_p+\F_p\begin{bmatrix}
        \F_{c-}\\ \F_{mk}
    \end{bmatrix}^\perp z_h\\
    &\eqqcolon\F_{pc}z_p+\F_hz_h.
    \end{split}
\end{equation}

% Furthermore, for a given past trajectory $\hat{w}_{c-}=\F_{c-}z_m$, all corresponding values of $w_m(k)$ can be similarly calculated as $w_m(k)=\F_{mk}\F_{c-}^\dagger\F_{c-}z_m+\F_{mk}\F_{c-}^\perp z_m$. Combining with \eqref{eq:zpCal} and using the facts that $\F_{c-}\F_{c-}^\dagger\F_{c-}=\F_{c-}$ and $\F_{c-}\F_{c-}^\perp=0$

% \begin{equation}
%     \hat{w}_p=\F_{cp}z_m+\F_{cm}z_m+\F_hz_h,
% \end{equation}
% where 
% \begin{equation*}
%     \begin{split}
%         \F_{cp}&=\F_p\begin{bmatrix}
%         \F_{c-}\\ \F_{mk}
%     \end{bmatrix}^\dagger\begin{bmatrix}
%         \F_{c-}\\ \F_{mk}
%     \end{bmatrix}\F_{c-}^\dagger\F_{c-},\\
%     \F_{cm}&=\F_p\begin{bmatrix}
%         \F_{c-}\\ \F_{mk}
%     \end{bmatrix}^\dagger\begin{bmatrix}
%         \F_{c-}\\ \F_{mk}
%     \end{bmatrix}\F_{c-}^\perp,\\
%     \F_h&=\F_p\begin{bmatrix}
%         \F_{c-}\\ \F_{mk}
%     \end{bmatrix}^\perp.
%     \end{split}
% \end{equation*}
The controlled behavior is decomposed into 2 components, with $\F_{pc}$ representing the past trajectories and current manipulated variable values, and $\F_h$ dealing with the behavior ``blocked'' by the network interconnection} For stand-alone systems, this concept reduces to the hidden behavior defined in \cite{Willems:2002}. \revise{With the presence of the blocked behavior, a trajectory of the controller network may correspond to a set of trajectories in the interconnected system that are indistinguishable from the viewpoint of controllers. A key issue in the verification of the existence of controlled behavior is hence to verify whether the desired controlled behavior exists that can accommodate the blocked behaviors. Furthermore, note that this decomposition does not restrict the behavior of $\hat{w}_p$ because $\cs(\begin{bmatrix}
    \F_{pc} & \F_c
\end{bmatrix})=\cs(\F_p)$.
\begin{exmp}[Example \ref{exmp:construction} continued]\label{exmp:hidden}
    For the system depicted in Fig. \ref{fig:intercon}, we assume that, while the current step disturbance $d(k)$ can be measured, the measurement is not timely enough to be used in the optimization of the $k$th step. In other words, $d(k)$ is \emph{unknown} in the $k$th optimization horizon. In such a case, predicted $d(k)$ (contained in $y_l^1(k)$) will most likely be different from the actual disturbance at the $k$th step, leading to a different set of outputs compared to the predicted ones, i.e., the current step disturbance is ``blocked'' by the network interconnection. Therefore, for disturbance attenuation, the goal should be to attenuate the effect of $d(k)$ even if its prediction is incorrect.
\end{exmp}
}

% Since $\B_p$ is time-invariant, the trajectory $\hat{w}_p$ can be indexed so that it is centered around the $k$th step, as shown in \eqref{eq:LTI+diss1}. Partitioning the interval $[k-L^-,k]$ as the ``past'' and $[k+1,k+L^+]$ as the ``future'', the coefficient matrix can be partitioned row-wise accordingly as
% \begin{subequations}\label{eq:pastfuture}
% 	\begin{align}
% 		\hat{w}_{p\bint{k-L^-,k}}&=\F _{p\bint{1,L^-+1}}z_{p_k},\\
% 		\hat{w}_{p\bint{k+1,k+L^+}}&=\F _{p\bint{L^-+2,L}}z_{p_k}.
% 	\end{align}
% \end{subequations}
% Similar to the construction of $\B_p$, the set-theoretic construction of the behavior of the interconnected controller behavior $\B_c$ can be constructed through the (to-be-designed) individual controller behaviors $\B_c^j$ and the controller network behavior $\B_c^\Pi$ given by 
% \begin{equation}\label{eq:controlnetwork}
% 	\Pi_cw_c=0.
% \end{equation}  
% The coefficient matrix for $L$-step trajectories of $w_c$ can be extended accordingly as $\kronI{\Pi}{c}{L}$. Similarly, the coefficient matrices in \eqref{eq:systemcontrollernetwork} can be extended accordingly to $\kronI{\Pi}{\revise{ip}}{L}$ and $\kronI{\Pi}{\revise{ic}}{L}$ for $L$-step trajectories of $w_p$ and $w_c$, respectively.

Let the desired control performance \revise{be} specified by the QdF $Q_{\Phi_d}(w_p)$. This condition can be integrated into the system as the virtual network interconnection with representation $s_d=Q_{\Phi_d}(w_p)$ and $V=Q_{\Psi}(w_p)$, where the coefficient matrix $\widetilde{\Psi}$ is to be constructed. The resulting largest possible desired controlled behavior, denoted as $\B_{pd}$, is of the form \eqref{eq:LTI+diss} with $\widetilde{\Phi}=\widetilde{\Phi}_d$. The aim of the control design is hence to construct the behaviors of the distributed controllers $\B_c^j$ such that $\B_{pc}\subset\B_{pd}$. \revise{Furthermore, the controlled behavior $\B_{pc}$ should be implementable through the distributed controllers with communication network interconnection $\B_c^\Pi$.} 

\revise{To ensure the feasibility of distributed optimization, we aim to find local conditions for each controller such that the satisfaction of these conditions implies the desired global performance condition described by $s_d$. Within each optimization horizon, each controller determines its current step action based on its past trajectory. Similar to the construction in \eqref{eq:wpControlled}, all possible trajectories of $\hat{w}_c^j$ based on its past can be expressed as
\begin{equation}\label{eq:controllerLocal}
    \hat{w}_c^j=\F_c^j\F_{c-}^{j\dagger}\F_{c-}^jz_p^j+\F_c^j\F_{c-}^{j\perp}z_m^j\eqqcolon\F_{cm}^jz_p^j+\F_{cf}^jz_m^j,
\end{equation}
where $\F_{c-}^j$ is defined similarly as $\F_{c-}$. We see that each choice of $z_p^j$ corresponds to a past trajectory while the choices of $z_m^j$ reveals different possible future outcomes by choosing different current step control actions. Similarly, for the entire network of distributed controllers, all possible trajectories based on a measured past are
\begin{equation}\label{eq:controllerGlobal}
    \hat{w}_c=\F_c\F_{c-}^{\dagger}\F_{c-}z_p+\F_c\F_{c-}^{\perp}z_m\eqqcolon\F_{cm}z_p+\F_{cf}z_m.
\end{equation}
If, for any $z_p$ and $\{z_p^j\}_{j=1}^{N_c}$, there exist $z_m$ and $\{z_m^j\}_{j=1}^{N_c}$ that lead to the desired controlled behavior, then this behavior is implementable by the distributed controllers. Note that, while any combination of $z_p$ and $z_m$ in \eqref{eq:controllerGlobal} leads to a corresponding value of $z_p$ in \eqref{eq:controlImp} and vise versa, the values of $z_p$ in these two equations are typically different for the same trajectory.
\begin{thm}\label{thm:controlledexist}
	For the setup given in Problem~\ref{prob:problem}, the desired controlled behavior exists and is implementable by the distributed controllers for the interval \revise{$[k-L^-,k]$} for any $k>L^-$, in which $L^->\max\{\lag{\B_p},\lag{\B_c^\Pi},\lag{\B_{pc}^\Pi},K_{\Phi_d}\}$, if $\mathrm{rank}\left(\F _{fk}\right)=\mathrm{w_{f}}$, where $\F_{fk}$ is the sub-matrix of $\F_p$ corresponding to $w_f(k)$, and there exist a $\widetilde{\Psi}\in\mathbb{S}^{L^-\mathrm{w_p}}$ and a set of matrices $\{\widetilde{\Phi}_c^j\}_{j=1}^{N_c}$, where $\widetilde{\Phi}_c^j\in\mathbb{S}^{(L^-+1)\mathrm{w}_\mathrm{c}^j}$, such that
	\begin{subequations}\label{eq:controlledexist}
	    \begin{align}
	        &\widetilde{\Psi}\geq0,\\
	        &\F_{cm}^{j\top}\widetilde{\Phi}_c^j\F_{cm}^{j}\geq0, \ \forall j\in\Zp{N_c},\label{eq:DOF}\\
	        &\begin{bmatrix}
	            \F_{pc}^\top\widetilde{\Delta}\F_{pc}-\F_{cc}^\top P^\top\widetilde{\Phi}_cP\F_{cc} &\F_{pc}^\top\widetilde{\Delta}\F_h\\
	            \F_h^\top\widetilde{\Delta}\F_{pc}& \F_h^\top\widetilde{\Delta}\F_h
	        \end{bmatrix}\geq0,\label{eq:minimum}\
	    \end{align}
	\end{subequations}
	where $\widetilde{\Delta}=\widehat{\Phi}_d-\nabla\widetilde{\Psi}$, $\widehat{\Phi}_d$ is defined similarly to $\widehat{\Phi}$ in Proposition \ref{prop:Hankeldiss}, $\widetilde{\Phi}_c=\diag\{\widetilde{\Phi}_c^j\}_{j=1}^{N_c}$, and $P$ is a permutation matrix such that $\col\{\hat{w}_c^j\}_{j=1}^{N_c}=P\hat{w}_c$.

\end{thm}
}
\begin{proof}
	See Appendix \ref{appx:proofthm:controlledexist}.
\end{proof}

% \begin{rem}\label{rem:fullrank}
% 	In practice, the matrix $\begin{bmatrix}
% 		\kronI{\Pi}{\revise{ic}}{L+1}\\\kronI{\Pi}{c}{L+1}
% 	\end{bmatrix}$ is often of full row rank. In such a case, we see that $\F _{cr}=0$, i.e., interconnection with the controller network poses no extra constraints on the system behavior.
% \end{rem}

\revise{In Theorem 12, \eqref{eq:minimum} implies the non-negativity of its second diagonal block, which means that the behavior blocked by the system-controller network interconnection is $(Q_{\Phi_d}(w_p),Q_{\Psi}(w_p))$-dissipative in order for the desired implementable controlled behavior to exist. This is a generalized version of the hidden behavior condition specified in \cite{Willems:2002} applied to interconnected systems. In addition, in order to ensure dissipativity of the controlled system, the controllers are required to satisfy a stronger condition than the required dissipativity (which only requires non-negativity of the first diagonal block of \eqref{eq:minimum}. This is because the controllers need to ensure $(Q_{\Phi_d}(w_p),Q_{\Psi}(w_p))$-dissipativity based on trajectories in \eqref{eq:wpControlled} with the lowest value of $Q_\Delta(w_p)$ for a given $z_p$ among all $z_h$, which can be achieved by enforcing a stricter dissipativity condition. The dissipativity of the blocked behavior implies the existence of a lower bound for $Q_\Delta(w_p)$ for all $z_h$.

The rank condition in Theorem \ref{thm:controlledexist} is only posed on the sub-matrix concerning the $k$th step of the chosen free variables in the controlled behavior ($w_f$). This ensures that $w_f(k)$ is free in the uncontrolled behavior, a prerequisite for $w_f$ to be free in the controlled behavior. On the other hand, \eqref{eq:DOF} ensures that, for any past trajectories, a $k$th step control action that is admissible through the controller network and renders the controlled system dissipative exists. Effectively, \eqref{eq:DOF} ensures that the past trajectories are free, other than that they must be trajectories of the behavior.} This requirement is due to the receding horizon of DDPC - \revise{to ensure recursive feasibility,} the controlled system cannot pose restriction on the trajectories that has already happened. \revise{As such, the entire trajectory $\hat{w}_f$ is free, achieving the second goal in Problem \ref{prob:problem}.}
\revise{\begin{rem} In Theorem~\ref{thm:controlledexist}, the  dissipativity condition \eqref{eq:controlledexist} is only required for the $k$th step of the trajectories in the controlled behavior for desired control performance/stability, which is actually implemented in every horizon.  In some cases, not ensuring the dissipativity condition for future steps may lead to unrealistic ``optimized'' control actions. Theorem \ref{thm:controlledexist} readily generalizes to including multiple future steps (See Theorem \ref{thm:controlledexistGeneral} in Appendix \ref{appx:controlGeneral}), which addresses this issue. However, such a generalization introduces conservatism because it requires dissipativity to be satisfied for the smallest $Q_\Delta(w_p)$ several steps into the future, a scenario that is unlikely to happen. A possible compromise could be to ensure dissipativity for the $k$th step (which is always necessary) and assume constant future trajectories of $w_f$ (which correspond to a subset of components in $w_c$) at the predicted $w_f(k)$, i.e., we predict $w_f(k)$ corresponding to the smallest $Q_\Delta(w_p)(k)$ and plan the future trajectory based on it. 
\end{rem}}

\subsection{DDPC Implementation Procedure}\label{sec:DDPC implement}
%In Theorem \ref{thm:controlledexist}, we have shown that condition \eqref{eq:conzp} needs to be implemented across several steps ($L^++1$ steps, to be exact). From \eqref{eq:zpnetworkcon}, we see that the network constraint has already be embedded in the new latent variable $\tilde{z}_{p_k}$. Therefore, when translating \eqref{eq:conzp} to have $\tilde{z}_{p_k}$ as the latent variable, we have that $\Theta_l=\F_{p_l}^\top \widetilde{\Delta}\F_{p_l}$ for all $i\in\Znn{L^+}$. Therefore, the transformed constraint of \eqref{eq:conzp} becomes
%\begin{equation}\label{eq:conzp2}
%	\tilde{z}_{p_k}^\top \widetilde{\F }_{p_l}^\top\left[\widetilde{\Delta}- \widetilde{\Delta} \F_{h_l}\left(\F_{h_l}^\top \widetilde{\Delta} \F_{h_l}\right)^{-1}\F_{h_l}^\top \widetilde{\Delta}\right]\widetilde{\F }_{p_l}\tilde{z}_{p_k}\geq0.
%\end{equation}
%where $\widetilde{\F}_{p_l}=\F_{p_l}(\widetilde{\F}_{cr}\F_p)^\perp$. To further reduce the condition for local distributed optimization, note that, for all $j\in\Zp{N_c}$, the transformation $\tilde{z}_{p_k}'=\mathcal{V}_c^{j\top}\tilde{z}_{p_k}$, where $\mathcal{V}_c^j$ is given in \eqref{eq:controlparticular}, defines a new set of latent variable for the system and contains the variable $\tilde{z}_{p_k}^j$. Furthermore, $\tilde{z}_{p_k}=\mathcal{V}_c^j\tilde{z}_{p_k}'$ because $\mathcal{V}_c^j$ is orthogonal.
\revise{With the conditions in Theorem \ref{thm:controlledexist} satisfied, local conditions $Q_{\Phi_c^j}(w_c^j)(k)\geq0$ are readily usable to achieve the desired global performance. However, since the controllers optimize for their trajectories in parallel, a concensus condition is required so that the resulting trajectory is admissible through the controller network and the system-controller network, i.e., it should be an admissible trajectory for the behavior represented by \eqref{eq:controlledLargest}. Following Lemma \ref{lem:Ax=b}, a consensus condition that guarantees existence of $z_{p_k}$, hence $\hat{w}_{p_k}$, is that 
\begin{equation}\label{eq:consesus}
    \begin{bmatrix}
        (\widetilde{\Pi}_{ip}\F_p)_\perp \widetilde{\Pi}_{ic}\\ \widetilde{\Pi}_c
    \end{bmatrix}\bar{w}_{c_k}=0,
\end{equation}
where $\bar{w}_{c_k}$ contains known past trajectory $\hat{w}_{c\bint{k-L^-,k-1}}$ and the to-be-determined future trajectory $\hat{w}_{c\bint{k,k+L^+}}$. The online optimization problem for the interval $[k-L^-,k+L^+]$ is therefore formulated as
\begin{equation}\label{eq:DDPC}
	\begin{split}
		\min_{z_{p_k}^1,z_{p_k}^2,\ldots,z_{p_k}^{N_c}} \ &\sum_{j=1}^{N_c}\sum_{m=0}^{L^+}{\revise{J^j}}\left(w_{c}^j(k+m)\right)\\
		\mathrm{s.t.} \ & \bar{w}_{c_k}^j=\F_c^jz_{p_k}^j,\ j\in\Zp{N_c},\\
		& \bar{w}_{c_0}^{j\top}\Phi_c^j\bar{w}_{c_0}^j\geq0, \ j\in\Zp{N_c},\\
		& \eqref{eq:consesus},
	\end{split}  
\end{equation}
where $\bar{w}_{c_k}^j$ is defined similarly to $\bar{w}_{c_k}$ and $\bar{w}_{c_0}^j$ contains past trajectory and to-be-determined $w_c^j(k)$. At the $k$th horizon, the distributed controllers optimize for the first two constraints in parallel and perform consensus to decide a trajectory $\hat{w}_c$. Each controller then implements $w_m^j(k)$.} As illustrated in Section \ref{sec:existence}, recursive feasibility of \eqref{eq:DDPC} is guaranteed upon satisfaction of \eqref{eq:DOF}. The online procedure is a standard distributed optimization problem with coupling consensus constraints, which can be solved using methods such as the classic ADMM \cite{Boyd:2011} or the recently developed ALADIN-$\alpha$ \cite{Engelmann:2020} and ELLADA \cite{Tang:2021}. \revise{The complexity of the local constraints depend on the number of variables in each subsystem and the value of $L^+$. However, since they are optimized in parallel, the complexity is not expected to increase rapidly as the scale of the problem increases. The consensus constraint, while higher in dimensions, is a linear constraint, and the optimization speed depends on the algorithms used.
\begin{rem} The behavioral system approach offers some benefit over model-based control with system identification for distributed control. An input-output model, by definition, describes the relationship between input and output for all possible input signals, i.e., the complete behavior of a system for the unrestricted interval $\mathbb{Z}^{\geq0}$. As such, a model of a subsystem interconnected with other subsystems viewed from its input-output variables will be complex, potentially including the dynamics of all subsystems, leading to a complex controller, e.g., requiring long optimization horizons (the worst case would be the sum of the lags of all subsystems). The behavioral approach, on the other hand, views the dynamics of a subsystem as the projection of the interconnected behavior on the restricted interval (e.g., $[1,L]$) onto the space of its variables. As such, as long as the interval is longer than the largest lag among all subsystems, the control design can be carried out entirely within the restricted interval.
\end{rem}}

\revise{\section{Numerical Example}\label{sec:example}
We continue the example depicted in Fig. \ref{fig:intercon} with transfer matrices of the subsystems as
\begin{subequations}\setlength{\arraycolsep}{1.1pt}
	\begin{align}
		Y^1&=\left[\begin{array}{cc|cc|c}
			\frac{1}{z+0.5}&0&\frac{0.2}{z+0.5}&0&\frac{0.8}{z+0.5}\\0&\frac{1}{z-0.1}&\frac{0.1}{z-0.1}&\frac{0.5}{z-0.1}&\frac{0.6}{z-0.1}
		\end{array}\right]\begin{bmatrix}
			U_p^1 \\U_c^1\\ D
		\end{bmatrix},\\
		Y^2&=\begin{bmatrix}
			0&-\frac{0.5}{z+0.3}\\\frac{z-0.34}{4z^2+1.2z-0.4}&-\frac{0.2}{z-0.2}
		\end{bmatrix}U_p^2,\\
		Y^3&=\left[\begin{array}{cccc|cc}
			\frac{1}{z-0.3}&0&\frac{0.2}{z-0.3}&\frac{0.1}{z-0.3}&\frac{0.5}{z-0.3}&0\\
			\frac{0.1}{z-0.3}&-\frac{0.54}{z-0.6}&\frac{z-0.34}{6.85z^2+6.17z+1.23}&\frac{0.01}{z-0.3}&\frac{0.05}{z-0.3}&\frac{0.9}{z-0.6}
		\end{array}\right]\begin{bmatrix}
		    U_p^3\\ U_c^3
		\end{bmatrix},\\
		Y^4&=\left[\begin{array}{cc|cc}
			\frac{0.1}{z-0.1}&\frac{0.2}{z-0.1}&\frac{1}{z-0.1}&0\\0&-\frac{0.5}{z+0.35}&0&\frac{1}{z+0.35}
		\end{array}\right]\begin{bmatrix}
			U_p^4 \\U_c^4
		\end{bmatrix}.
	\end{align}
\end{subequations}
Note that these models are only used to generate the data to construct Hankel matrices. They are not involved in either offline design or online implementation.

The distributed controllers aim to attenuate the effect of disturbance $d$ to $y^3$ and $y^4$ while maintaining the stability of the outputs of the other subsystems. Specifically, we aim to attenuate the effect of disturbance to $y^3$ and $y^4$ by 20 times at steady state and the attenuation is achieved within approximately 25 steps. In discrete time, this performance requirement can be constructed in terms of a frequency-weighted $\pazocal{H}_\infty$ condition $\|WT_{yd}\|_\infty\leq1$, where $T_{yd}$ is the system from $d$ to $\col(y^3,y^4)$, $W(\sigma)=\frac{3.625\sigma}{\sigma-0.8187}I_4$ is a weighting function that specifies the performance requirements. As illustrated in \cite{Yan:2019}, the weighting function readily transforms into a QdF supply rate that automatically guarantees the stability of the rest of the variables. We omit the details of constructing the QdF due to the page limit. 

As discussed in Example \ref{exmp:hidden}, we assume that the measurement of $d(k)$ is not included during online optimization. In each horizon, we choose $L^-=4$, which is larger than the lags of all subsystems and the network interconnection. Since the main purpose of this example is to demonstrate the use of dissipativity, we focus on the interval $[k-4,k]$ for the $k$th horizon, i.e., $L^+=0$. Online optimization in each horizon determines $y_l^j(k)$, $y_r^j(k)$, $u_l^j(k)$ and $u_r^j(k)$ based on their past measured trajectories and implement $y_l^j(k)$ (i.e., $u_c^i(k)$ from the perspective of the interconnected system). The stage cost in \eqref{eq:DDPC} is chosen as
\begin{equation}
    J^j=\|y_l^j(k)\|^2,
\end{equation}
which aims to minimize the control effort and help avoid the potential controller instability problem as illustrated in \cite{Markovsky:2008}.

The local constraints are found by solving the LMIs in Theorem \ref{thm:controlledexist}. Online DDPC optimization \eqref{eq:DDPC} is carried out with disturbance profile shown in the first plot of Fig. \ref{fig:plot}. As is shown in the plots of $y^3$ and $y^4$, the effect of disturbance is attenuated within the required number of steps with every change in the disturbance. Furthermore, stability of $y^1$ and $y^2$ are guaranteed with the presence of disturbance.
\begin{figure}
		\centering
		\includegraphics[width=0.8\linewidth]{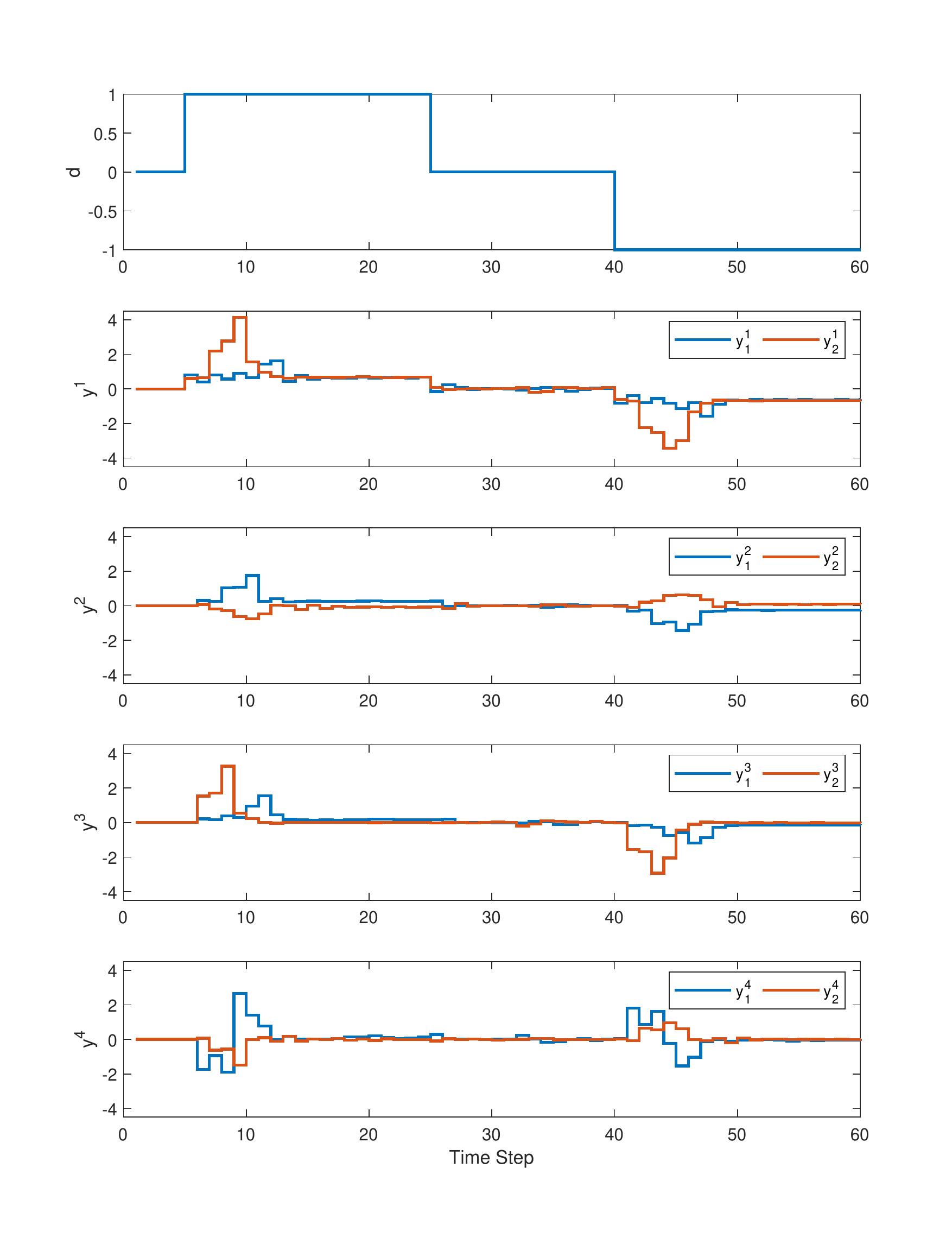}       
		\caption{Disturbance Profile and Controlled Outcome}
		\label{fig:plot}
	\end{figure}}
\section{Conclusion}\label{sec:conclusion}
In this paper, we have developed a structure for Distributed Data-driven Predictive Control of interconnected systems with LTI subsystems. The entire design procedure is carried out in the behavioral framework, which has been shown to facilitate data-driven control design. We have shown that, like the case with infinite time axis, the complete finite-length interconnected behavior can be obtained from that of the subsystems and of the network interconnection, provided that the number of steps is larger than the maximum lag among all subsystems. Furthermore, we have given some treatments to dissipativity from the behavioral perspective and have shown that they can be analyzed and handled like any other types of dynamical systems. \revise{Based on these components, a set of conditions that verifies the existence of desired controlled behavior has been developed, leading to a DDPC algorithm that achieves the global design goal through distributed optimization.}

Future research includes the extension of the above results to handle the noise in measurements. This is a challenging but important issue for data-driven control because it affects both offline design (similar to a model-plant mismatch problem) and online implementation. While unstructured low-rank approximation is a promising way to approximately parameterize the trajectory space \cite{Markovsky:2012}, how dissipativity can be modified to compensate for the error of approximation in the subsystems that may be exacerbated by the complex interconnection is an interesting question worth further investigation.

\appendices
\revise{\section{General Results for Dissipativity Synthesis}\label{appx:controlGeneral}
This appendix presents the general results to Section \ref{sec:controldesgin} with the assumptions in Section \ref{sec:preamble} relaxed. With the general case, the largest behavior for the controllers to choose from needs to be constructed by viewing the behavior of the entire controlled system (i.e., system depicted in Fig. \ref{fig:fig-plantwide}) as a latent variable behavior with manifest variable $w_c$ and latent variable $w_p$ (or $z_p$ to be more direct) and finding its manifest behavior. The manifest behavior can be constructed through the following proposition, whose proof is omitted because it is a direct consequence from the definition of manifest behavior given by \eqref{eq:manifestbehavior} and the application of Lemma \ref{lem:Ax=b}.
\begin{prop}\label{prop:manifestLstep}
	Let $\Sigma^{full}\in\LTI{\mathrm{w+l}}$ with $\B^{full}_\bint{1,L}$ represented by  $\widetilde{R}\hat{w}=\widetilde{M}\hat{\ell}$, then 
	\begin{equation}
	    \B_\bint{1,L}=\cs\left((\widetilde{M}_\perp\widetilde{R})^\perp\right).
	\end{equation}
\end{prop}

From Proposition \ref{prop:manifestLstep}, all possible trajectories for $w_c$ before dissipativity synthesis are given by
\begin{equation}\label{eq:controllerManifestGeneral}
    \hat{w}_c=\begin{bmatrix}
        (\widetilde{\Pi}_{ip}\F_p)_\perp \widetilde{\Pi}_{ic}\\ \widetilde{\Pi}_c
    \end{bmatrix}^\perp z_c\eqqcolon\F_cz_c.
\end{equation}
While we use the same notations as those in Section \ref{sec:controldesgin} for various behaviors in this appendix, the coefficient matrices are not the same. Note that, different from the special case, in which $z_p$ in \eqref{eq:controllerManifest} shares the same variable as \eqref{eq:systemManifest}, the variable $z_c$ \eqref{eq:controllerManifestGeneral} is generally different from $z_p$ for a corresponding pair of trajectories $(\hat{w}_p,\hat{w}_c)$ admissible through the network.

Following similar rationale in the construction of \eqref{eq:controlImp}, an implemented trajectory can be represented by $\hat{w}_c=\F_{cc}z_c$. However, to obtain corresponding trajectories of $\hat{w}_p$, the implemented trajectory needs to be integrated directly into \eqref{eq:controlledLargest} to obtain the corresponding values of $z_p$. By Lemma \ref{lem:Ax=b}, $z_c$ constructed in this way definitely has solutions for $z_p$, and all of its solutions can be represented as
\begin{equation}
    z_p=(\widetilde{\Pi}_{ip}\F_p)^\dagger\widetilde{\Pi}_{ic}\F_{cc}z_c+(\widetilde{\Pi}_{ip}\F_p)^\perp z_h,
\end{equation}
leading to a general representation of all trajectories of $\hat{w}_p$ corresponding to $z_c$ as
\begin{equation}\label{eq:wpControlledGeneral}
\begin{split}
    \hat{w}_p&=\F_p(\widetilde{\Pi}_{ip}\F_p)^\dagger\widetilde{\Pi}_{ic}\F_{cc}z_c+\F_p(\widetilde{\Pi}_{ip}\F_p)^\perp z_h\\
    &\eqqcolon \F_{pc}z_c+\F_hz_h.
\end{split}
\end{equation}
In this general version, it is clear that the second component $\F_h$ represents the behavior ``blocked'' by the network interconnection because $\widetilde{\Pi}_{ip}\F_h=0$. Furthermore, while \eqref{eq:wpControlled} does not eliminate any trajectories from $\B_p$, there may be eliminations in \eqref{eq:wpControlledGeneral} because $\F_c$ in \eqref{eq:controllerManifestGeneral} may only parameterize a subset of trajectories represented by the column span of $\F_c$ in \eqref{eq:controllerManifest}. 

Suppose the aim is to ensure dissipativity up to step $k+L_p^+$ ($L_p^+\leq L^+)$. Similar to \eqref{eq:PhiLdissSOS}, we denote the segment of the coefficient matrix corresponding to trajectories in the interval $[k-L^-+l,k+l]$ with an additional subscript $l$, e.g., $\F_{c_l}$, $\F_{c_l}^j$, etc. The general result to Theorem \ref{thm:controlledexist} is then as follows.
\begin{thm}\label{thm:controlledexistGeneral}
    The desired implementable controlled behavior exists for the interval $[k-L^-,k+L_p^+]$ for any $k>L^-$, where $L^->\max\{\lag{\B_p},\lag{\B_c^\Pi},\lag{\B_{pc}^\Pi},K_{\Phi_d}\}$, if $\mathrm{rank}\left(\F _{fp}\right)=(L_p^++1)\mathrm{w_{f}}$, where $\F_{fp}$ is the sub-matrix of $\begin{bmatrix}
        \F_{pc} & \F_h
    \end{bmatrix}$ corresponding to $\hat{w}_{f\bint{k,k+L_p^+}}$, and there exist a $\widetilde{\Psi}\in\mathbb{S}^{L^-\mathrm{w_p}}$ and a set of matrices $\{\widetilde{\Phi}_c^j\}_{j=1}^{N_c}$, where $\widetilde{\Phi}_c^j\in\mathbb{S}^{(L^-+1)\mathrm{w}_\mathrm{c}^j}$, such that $\widetilde{\Psi}\geq0$, and , for all $l\in\Znn{L_p^+}$,
	\begin{subequations}\setlength{\arraycolsep}{2pt}\label{eq:controlledexistGeneral}
	    \begin{align}
	        &z_c^{j\top}\F_{cm_l}^{j\top}\widetilde{\Phi}_c^j\F_{cm_l}^{j}z_c^j\geq0, \ \forall j\in\Zp{N_c},\label{eq:DOFGeneral}\\
	        &z_{s_l}^\top{\footnotesize\begin{bmatrix}
	            z_c^\top(\F_{pc_l}^\top\widetilde{\Delta}\F_{pc_l}-\F_{cc_l}^\top P^\top\widetilde{\Phi}_cP\F_{cc_l})z_c &z_c^\top\F_{pc_l}^\top\widetilde{\Delta}\F_{h_l}\\
	            \F_{h_l}^\top\widetilde{\Delta}\F_{pc_l}z_c& \F_{h_l}^\top\widetilde{\Delta}\F_{h_l}
	        \end{bmatrix}}z_{s_l}\geq0,\label{eq:minimumGeneral}\
	    \end{align}
	\end{subequations}
	for all $\{z_c^j\}_{j=1}^{N_c}$, $z_c$ and $\{z_{s_l}\}_{l=0}^{L_p^+}$, where $\widetilde{\Delta}$ and $\widetilde{\Phi}_c$ are defined in Theorem \ref{thm:controlledexist}, $P$ is a matrix such that $\col\left\{\hat{w}_{c\bint{k-L^-,k}}^j\right\}_{j=1}^{N_c}=P\hat{w}_{c\bint{k-L^-,k}}$ and $\{z_{s_l}\}_{l=0}^{L_p^+}$ is a set of ($l$-dependant) auxiliary variables.
\end{thm}

The proof of Theorem \ref{thm:controlledexistGeneral} follows similar rationale as that of Theorem \ref{thm:controlledexist}, so we only highlight the main differences. Firstly, the rank condition is generalized because, as explained above, \eqref{eq:wpControlledGeneral} may only be a subset of $\cs(\F_p)$. The rank condition guarantees that, while there may be restrictions on $\hat{w}_p$, no restrictions can be posed on $\hat{w}_{f\bint{k,k+L_p^+}}$. Furthermore, the aim is still to achieve \eqref{eq:disscomplete}, but for intervals $[k-L^-+l,k+l]$ for all $l\in\Znn{L_p^+}$ with the \emph{same} $z_p$, as required by Proposition \ref{prop:Hankeldiss}. Similarly, conditions in \eqref{eq:controlledexistGeneral} are generalized not only with the general behavior \eqref{eq:controllerManifestGeneral}, but also across multiple segments with the same set of variables. With the help of the auxiliary variables $\{z_{s_l}\}_{l=0}^{L_p^+}$ (and another variable applied similarly to $\widetilde{\Psi}$), conditions in Theorem \ref{thm:controlledexistGeneral} are transformed into a sum-of-square (SOS) problem, which is solvable using existing numerical toolboxes such as YALMIP \cite{Lofberg:2004}. The online optimization problem \eqref{eq:DDPC} can be modified accordingly as
\begin{equation}\label{eq:DDPCGeneral}
	\begin{split}
		\min_{z_{c_k}^1,z_{c_k}^2,\ldots,z_{c_k}^{N_c}} \ &\sum_{j=1}^{N_c}\sum_{m=0}^{L^+}{\revise{J^j}}\left(w_{c}^j(k+m)\right)\\
		\mathrm{s.t.} \ & \bar{w}_{c_k}^j=\F_c^jz_{p_k}^j,\ j\in\Zp{N_c},\\
		& \bar{w}_{c_l}^{j\top}\Phi_c^j\bar{w}_{c_l}^j\geq0, \ j\in\Zp{N_c}, \ l\in\Znn{L_p^+}\\
		& \eqref{eq:consesus},
	\end{split}  
\end{equation}
where $\bar{w}_{c_l}^j$ contains the past trajectory $\hat{w}_{c\bint{k-L^-+l,k-1}}^j$ and the to-be-determined future trajectory $\hat{w}_{c\bint{k,k-l}}^j$. 
}
\section{Proof of Results}
Many of the proofs rely on the following result:
\begin{lem}[\cite{Skelton:1997}]\label{lem:Ax=b}
	The equation $Ax=b$, where $A\in\R^{\bullet\times\mathrm{n}}$, $x\in\R^\mathrm{n}$ and $b\in\R^{\bullet}$ has solutions for $x$ if and only if $A_\perp b=0$. In such a case, all solutions are given by $x=A^\dagger b+A^{\perp}z$, where $z\in \R^\mathrm{n}$ is arbitrary.
\end{lem}

\subsection{Proof of Theorem \ref{thm:behaviorPWinvariant}}\label{appx:proofthm:behaviorPWinvariant}
To prove this theorem, we need the auxiliary results below.
\begin{prop}\label{prop:truncintersect}
	Let $\B_1,\B_2\subset\W^\T$ be two time-invariant behaviors with lags $\lag{\B_1}$ and $\lag{\B_2}$. Then
	\begin{equation}\label{eq:intersectinclude}
		\left(\B_1\cap\B_2\right)_{\bint{1,L}}\subset\B_{1\bint{1,L}}\cap\B_{2\bint{1,L}}.
	\end{equation}
	Furthermore, if $L>\max\left\{\lag{\B_1},\lag{\B_2}\right\}$, then
	\begin{equation}\label{eq:intersectequal}
		\left(\B_1\cap\B_2\right)_{\bint{1,L}}=\B_{1\bint{1,L}}\cap\B_{2\bint{1,L}}.
	\end{equation}
\end{prop}
\begin{proof}
	The proof requires two auxiliary results, which are stated in the following lemmas.
	\begin{lem}[\cite{Markovsky:2020}]\label{lem:truncextension}
		Let $\B\in\LTI{\bullet}$, then $\B_\bint{1,L}$ uniquely extends to $\B$ if $L>\lag{\B}$.
	\end{lem}
	\begin{lem}[\cite{Markovsky:2020}]\label{lem:truncagree}
		Let $\B_1,\B_2\in\LTI{\bullet}$ and let $L>\max\{\lag{\B_1},\lag{\B_2}\}$, then $\B_{1\bint{1,L}}=\B_{2\bint{1,L}}$ if and only if $\B_1=\B_2$.
	\end{lem}
	We are now in the position to prove this proposition. The behaviors in \eqref{eq:intersectinclude} are as follows
	\begin{align*}
		\left(\mathfrak{B}^1\cap\mathfrak{B}^2\right)_{|[1,L]}&=\bigl\{w|\exists w'\in \mathfrak{B}^1\cap\mathfrak{B}^2, \hat{w}=\hat{w}'_\bint{1,L}\bigr\},\\
		\mathfrak{B}^1_{|[1,L]}\cap\mathfrak{B}^2_{|[1,L]}&=\bigl\{w|\exists w_1'\in \mathfrak{B}^1, w_2'\in \mathfrak{B}^2,\\
		&\qquad\qquad\qquad\hat{w}=\hat{w}'_{1\bint{1,L}}=\hat{w}_{2\bint{1,L}}\bigr\}.
	\end{align*}
	We see that the latter is more relaxed than the former; hence we establish \eqref{eq:intersectinclude}. To show \eqref{eq:intersectequal}, note that, according to Lemma \ref{lem:interconproperty}(i), $L>\max\left\{\lag{\B_1},\lag{\B_2}\right\}\geq\lag{\B_1\cap\B_2}$. Therefore, according to Lemma \ref{lem:truncextension}, the LHS of \eqref{eq:intersectequal} uniquely extends to $\B_1\cap\B_2$. On the other hand, by similar reasoning, $\B_{1\bint{1,L}}$ and $\B_{2\bint{1,L}}$ uniquely extend to $\B_1$ and $\B_2$, respectively. Since both sides extend to the same behavior, they coincide according to Lemma \ref{lem:truncagree}. 
\end{proof}
\begin{prop}\label{prop:trunccartesian}
	Let $\B^1\subset\left(\W^1\right)^\T$ and $\B^2\subset\left(\W^2\right)^\T$ be two behaviors, then
	$ \left(\B^1\times\B^2\right)_\bint{1,L}=\B^1_\bint{1,L}\times\B^2_\bint{1,L}$.
	\end{prop}

The proof is similar to the first part of the proof of Proposition \ref{prop:truncintersect} and is therefore omitted.
\begin{prop}\label{prop:truncproj}
	Let $\B\in(\W_1\times\W_2)^\T$, then
	\begin{equation}
	    \proj{w_1}{\B}_\bint{1,L}\subset\proj{w_1}{\B_\bint{1,L}}.
	\end{equation}
\end{prop}
\begin{proof}
	The two behaviors in question are
	\begin{align*}
		\proj{w}{\mathfrak{B}}_{|[1,L]}&=\bigl\{w|\exists w',\ell, (w',\ell)\in\mathfrak{B},\hat{w}=\hat{w}'_\bint{1,L}\bigr\},\\
		\proj{w}{\mathfrak{B}_{|[1,L]}}&=\bigl\{w|\exists w',\ell,\ell',(w',\ell')\in\mathfrak{B},\\
		&\qquad\qquad\qquad\quad\hat{w}=\hat{w}'_\bint{1,L},\hat{\ell}=\hat{\ell}'_\bint{1,L}\bigr\}.
	\end{align*}
	\revise{After appropriate permutations to align the variable $w$ in the above equations,} the former is the latter with $\ell=\ell'$.
\end{proof}

We now prove Theorem~\ref{thm:behaviorPWinvariant}. On the one hand, by the combination of Lemma \ref{lem:interconbehaviourproj}(ii) and the propositions above, 
\begin{equation}\label{eq:interconsupset}
	\begin{split}
		&\left(\bigtimes_{i=1}^N\proj{w^i}{\B_\bint{1,L}}\right)\cap\B^\Pi_\bint{1,L}\\
		\supset&\left(\bigtimes_{i=1}^N\proj{w^i}{\B}_\bint{1,L}\right)\cap\B^\Pi_\bint{1,L}\\
		=&\left(\bigtimes_{i=1}^N\proj{w^i}{\B}\right)_\bint{1,L}\cap\B^\Pi_\bint{1,L}\\
		\supset&\left(\left(\bigtimes_{i=1}^N\proj{w^i}{\B}\right)\cap\B^\Pi\right)_\bint{1,L}=\B_\bint{1,L}.
	\end{split}
\end{equation}
On the other hand, according to Lemma \ref{lem:interconproperty}(i), $L>\lag{\B}$. Therefore, according to Lemma \ref{lem:interconbehaviourproj} and Proposition \ref{prop:truncintersect}, we have
\begin{equation}\label{eq:interconsubset}
	\begin{split}
		&\left(\bigtimes_{i=1}^N\proj{w^i}{\B_\bint{1,L}}\right)\cap\B^\Pi_\bint{1,L}\\
		\subset&\left(\bigtimes_{i=1}^N\B^i_\bint{1,L}\right)\cap\B^\Pi_\bint{1,L}
		=\left(\bigtimes_{i=1}^N\B^i\right)_\bint{1,L}\cap\B^\Pi_\bint{1,L}\\
		=&\left(\left(\bigtimes_{i=1}^N\B^i\right)\cap\B^\Pi\right)_\bint{1,L}=\B_\bint{1,L}.
	\end{split}
\end{equation}
The combination of \eqref{eq:interconsupset} and \eqref{eq:interconsubset} gives \eqref{eq:interconinvariant}.

\subsection{Proof of Theorem \ref{thm:behaviorPW}}\label{appx:proofthm:behaviorPW}
To begin with, viewing the interconnected system $\Sigma$ as a stand-alone system, the collection of all free variables in the subsystems are the free variables for $\Sigma$. \revise{Since the lower bound of $L$ is specified as the upper bound of $\mathtt{L}(\B)$ in Lemma \ref{lem:interconproperty}, we have $L>\lag{\B}$. Similar argument shows that the free variables are definitely persistently exciting of order no less than $L+\n{\B}$.} According to Lemma \ref{lem:behaviorpara}, a mosaic Hankel matrix $\Hk_L(\revise{\mathcal{W}})$ constructed from a set of trajectories $\mathcal{W}\subset\B_\bint{1,T}$ with its free component collectively persistently exciting of order $L+\n{\B}$ parameterizes $\B_\bint{1,L}$. It is easy to see that the sub-matrix of $\Hk_L(\revise{\mathcal{W}})$ containing the rows of trajectories of $w^i$ forms the mosaic Hankel matrix $\Hk_L(\revise{\mathcal{W}}^i)$, hence it parameterizes $\proj{w^i}{\B_\bint{1,L}}$. We therefore have \eqref{eq:Hankelproj}. 

Now, the behavior represented by 
\begin{subequations}\label{eq:fullintercon}
	\begin{align}
		\hat{w}&=\Hk_{L}g,\label{eq:fullintercon1} \\
		\widetilde{\Pi}\hat{w}&=0,\label{eq:fullintercon2}
	\end{align}
\end{subequations}
where $g\in\R^\bullet$ is an arbitrary vector, is precisely the LHS of \eqref{eq:interconinvariant}. Since $L>\lag{\B}$, from Theorem \ref{thm:behaviorPWinvariant}, \eqref{eq:fullintercon} represents the behavior $\B_\bint{1,L}$. Substituting \eqref{eq:fullintercon1} into \eqref{eq:fullintercon2} gives
\begin{equation}
	\widetilde{\Pi}\Hk_{L}g=0 \ \Rightarrow \ g=(\widetilde{\Pi}\Hk_{L})^\perp z,
\end{equation}
where $z\in\R^\bullet$ is arbitrary. Substituting into \eqref{eq:fullintercon1} gives \eqref{eq:Msystem}.
\subsection{Proof of Proposition \ref{prop:dissipativity}}\label{appx:proofprop:dissipativity}
The proof is inspired by the continuous-time counterpart in \cite{Willems:2007a}. For the \emph{if} part, summing both sides of \eqref{eq:dissstep} from $k_0$ to $k_1$ gives 
\begin{equation}\label{eq:disssum}
	-\sum_{k=k_0}^{k_1}s(k)\leq V(k_0-1)-V(k_1)\leq V(k_0-1).
\end{equation}
We therefore get \eqref{eq:diss} with $C=V_(k_0-1)$. For the \emph{only if} part, choose
\begin{equation*}
	V(k)=\sup_{K\geq k} \widetilde{V}_K(k), \ \widetilde{V}_k(k)=\begin{cases}
		 -\sum_{i=k+1}^{K}s(i), \ &\text{if } K>k,\\
		0, &  \text{if } K=k
	\end{cases}.
\end{equation*}
Note that $0\leq V(k)<\infty$ for all $k$. Then, for any $k\in\Znn{}$, if the supremum occurs at $K>k$, we have
\begin{align*}
	V(k)-V(k-1)&=\sup_{K}\left\{-\sum_{i=k+1}^{K}s(i)\right\}-\sup_{K}\left\{-\sum_{i=k}^{K}s(i)\right\}\\
	&\leq\sup_{K}\left\{\sum_{i=k}^{K}s(i)-\sum_{i=k+1}^{K}s(i)\right\}=s(k).
\end{align*}
On the other hand, if the supremum occurs at $K=k$, then $V(k)=0$ and $V(k-1)=-s(k)$, we immediately get \eqref{eq:dissstep}.

\subsection{Proof of Proposition \ref{prop:Hankeldiss}}\label{proofprop:Hankeldiss}
We need a technical tool to prove this proposition. This is stated in the following lemma.
\begin{lem}[\cite{Yan:2021}]\label{lem:projinout}
	Let $\Sigma=(\Sigma^1\sqcap\Sigma^2)\wedge\Sigma^\Pi$, then
	\begin{equation*}
		\proj{w^1}{(\B^1\times\B^2)\cap\B^\Pi}=\B^\revise{1}\cap\proj{w^1}{\left[(\W^1)^\T\times\B^2\right]\cap\B^\Pi}
	\end{equation*}
\end{lem}

We are now ready to prove the proposition. Firstly, condition (ii) ensure that, within the $L$ step, \eqref{eq:LTI+diss2} holds for all trajectories such that \eqref{eq:LTI+diss1} holds. In terms of the behaviors, \eqref{eq:LTI+diss2} represents the behavior $\B_{sw}=\proj{w_p}{(\W_p^\T\times\B_s)\cap\B_{ps}^\Pi}$ and the assumption in the proposition means that the lag of this behavior is bounded by $L^-$. Therefore, condition (ii) guarantees that $\B_{p\bint{k-L^-,k+L^+}}\subset\B_{sw\bint{k-L^-,k+L^+}}$, or
\begin{align*}
	\B_{p\bint{k-L^-,k+L^+}}=\B_{p\bint{k-L^-,k+L^+}}\cap\B_{sw\bint{k-L^-,k+L^+}}.
\end{align*}
From condition (i), we have $L>\max\{\lag{\B_p},\lag{\B_{sw}}\}$, which, according to Proposition \ref{prop:truncintersect}, means that the RHS of the above equation can be combined, i.e.,
\begin{align*}
	\B_{p\bint{k-L^-,k+L^+}}&=(\B_p\cap\B_{sw})_\bint{k-L^-,k+L^+}\\
	&=\proj{w_p}{(\B_p\times\B_s)\cap\B_{ps}^\Pi}_\bint{k-L^-,k+L^+}\\
	&=\proj{w_p}{\B_{ps}}_\bint{k-L^-,k+L^+},
\end{align*}
where the second line of equality uses Lemma \ref{lem:projinout}. Finally, since the lags of both sides are less than $L$ according to the first condition and \eqref{eq:lagdissproj}, we have the equality with unrestricted interval according to Lemma \ref{lem:truncagree}. 

\revise{In the case when $L^+=0$, \eqref{eq:PhiLdissSOS} reduces to only one inequality, i.e., $z_k^\top\F^\top (\widehat{\Phi}-\nabla\widehat{\Psi})\F z_k\geq0$ for all $z_k$, which is equivalent to \eqref{eq:PhiLdissLMI}.}

\subsection{Proof of Proposition \ref{prop:dissmani}}\label{appx:proofprop:dissmani}
We construct the manifest behavior based on the signature of $R$. Perform eigendecomposition on $R$ as
\begin{equation}\label{eq:EVD}
	R=V\Lambda V^\top\coloneqq\begin{bmatrix}
		V_+^\top \\ V_0^\top \\ V_-^\top 
	\end{bmatrix}^\top \begin{bmatrix}
		\Lambda_+&0&0\\ 0 & 0&0\\ 0& 0 & -\Lambda_-
	\end{bmatrix}\begin{bmatrix}
		V_+^\top \\ V_0^\top \\ V_-^\top 
	\end{bmatrix},
\end{equation} 
where $V_+$, $V_0$ and $V_-$ are conformable partitions of $V$ to those of $\Lambda$. 

\noindent \textbf{(i) $\boldsymbol{R\leq0}$}

\noindent Firstly, differentiating LHS of \eqref{eq:dissmanifest} and setting it to 0 lead to the equality
\begin{equation*}\label{eq:max}
	S^\top \hat{w}_1+R\hat{w}_2=0.
\end{equation*}
According to Lemma \ref{lem:Ax=b}, this equation has solutions for $\hat{w}_2$ if and only if $R_\perp S^\top\hat{w}_1=0$. If this is the case, then all solutions correspond to the maximum of \eqref{eq:dissmanifest} with respect to $\hat{w}_2$, which can be computed as
\begin{equation}
	\begin{split}
		&\hat{w}_1^\top Q\hat{w}_1+\hat{w}_1^\top S\hat{w}_2+\hat{w}_2^\top S^\top \hat{w}_1+\hat{w}_2^\top R\ell\\
		=&\hat{w}_1^\top Q\hat{w}_1-\hat{w}_2^\top R\hat{w}_2+\hat{w}_2^\top (S^\top \hat{w}_1+R\hat{w}_2)\\
		=&\hat{w}_1^\top Q\hat{w}_1-\hat{w}_2^\top RR^{\dagger} R\hat{w}_2\\
		=&\hat{w}_1^\top Q\hat{w}_1-\hat{w}_1^\top SR^{\dagger} S^\top \hat{w}_1.
	\end{split}
\end{equation}
If $R_\perp S^T\hat{w}\neq0$, then, for any given $\hat{w}$, choose
\begin{equation}
	\hat{w}_2=\frac{1}{2}\rho R_\perp S^\top\hat{w}_1,
\end{equation}
where $\rho\in\R$. \eqref{eq:dissmanifest} then becomes
\begin{equation*}
	\begin{split}
		&\hat{w}_1^\top Q\hat{w}_1+\rho\hat{w}_1^\top S R_\perp S^\top\hat{w}_1+\frac{\rho^2}{4}\hat{w}_1^\top SR_\perp RR_\perp S^\top\hat{w}_1\\
		=&\hat{w}_1^\top (Q+\rho S R_\perp S^\top)\hat{w}_1,
	\end{split}
\end{equation*}
which can always be made positive for appropriate choice of $\rho$ because $R_\perp\geq0$. In other words, $w_1$ is free. Note that this includes $R=0$ as a special case.

\noindent \textbf{(ii) $\boldsymbol{R\geq0}$}

\noindent We firstly show the case when $R>0$. Completing the square for \eqref{eq:dissmanifest} with respect to $w_1$ gives
\begin{equation*}
	(S^\top \hat{w}_1+R\hat{w}_2)^\top R^{-1}(S^\top \hat{w}_1+R\hat{w}_2)\geq \hat{w}_1^\top (SR^{-1}S^\top -Q)\hat{w}_1.
\end{equation*}
This is automatically true if the RHS is negative. If the RHS is non-negative, a possible set of solutions for $\hat{w}_2$ can be constructed as
\begin{equation}\label{eq:latentsol}
	\hat{w}_2=\rho R^{\frac{1}{2}}\sqrt{\hat{w}_1^\top (SR^{-1}S^\top -Q)\hat{w}_1}-R^{-1}S^\top \hat{w}_1,
\end{equation}
where $\rho\geq1$. In other words, there is always a corresponding $\hat{w}_2$ for any $\hat{w}_1$, hence $\hat{w}_1$ is free. If $R\geq0$, then define $\hat{w}_2'=V^\top\hat{w}_2=\col(\hat{w}_{2+}',\hat{w}_{20}')$, where $V$ is the eigenvector matrix in \eqref{eq:EVD}. We then have 
\begin{equation*}
	\hat{w}_2^\top R\hat{w}_2=(\hat{w}_{2+}')^\top \Lambda_+\hat{w}_{2+}'.
\end{equation*}
Since $\Lambda_+>0$, we can arbitrarily choose $\hat{w}_{20}'$ and compute $\hat{w}_{2+}'$ in a similar way as \eqref{eq:latentsol} by replacing $\hat{w}_1\rightarrow\col(\hat{w}_1,\hat{w}_{20}')$, $Q\rightarrow \begin{bmatrix}
	Q & SV_0\\ V_0^\top S^\top & 0
\end{bmatrix}$, $S\rightarrow\begin{bmatrix}
SV_+ \\ 0
\end{bmatrix}$, $R\rightarrow\Lambda_+$. We therefore obtain an $\hat{w}_2'$, hence $\hat{w}_2$, for any $\hat{w}_1$. Note that this case does \emph{not} include the case when $R=0$.

\noindent \textbf{(iii) $\boldsymbol{R}$ is indefinite}

\noindent By defining $\hat{w}_2'=V^\top \hat{w}_2=\col(\hat{w}_{2+}',\hat{w}_{20}',\hat{w}_{2-}')$, and choosing $\hat{w}_{2-}'=0$, this case reduces to Case (ii) and we conclude that $w_1$ is free.
%\subsection{Proof of Proposition \ref{prop:settoset}}
%To prove this proposition, we need an auxiliary result, which is stated in the following Lemma.
%\begin{lem}
%	Let $\B\in\LTI{\bullet}$. If $(w_1,w_2)\in\B$, $(w_1',w_2)\in\B$ and $(w_1,w_2')\in\B$, then $(w_1',w_2')\in\B$.
%\end{lem}
%\begin{proof}
%	Due to linearity, we have
%	\begin{align*}
%		(w_1,w_2)\in\B,(w_1',w_2)\in\B&\Rightarrow (w_1'-w_1,0)\in\B,\\
%		(w_1,w_2)\in\B,(w_1,w_2')\in\B&\Rightarrow (0,w_2'-w_2)\in\B.
%	\end{align*}
%	We therefore have $(w_1'-w_1,w_2'-w_2)\in\B$. Adding it to $(w_1,w_2)$ gives the result.
%\end{proof}
%
%Suppose that $\B_{xi}\neq\varnothing$, then consider trajectories $(w_p,w_c)\in\B_{xi}$, $(w_p,w_c')\in\B_{ex}$ and $(w_p',w_c')\in\B_{out}$. All three trajectories belong to the set
%\begin{equation*}
%	\B_{nd}=\left[\left(\B_p\setminus\proj{w_p}{\B_{in}}\right)\times\left(\B_c^\Pi\cap\B_{cr}\right)\right]\cap\B_{pc}^\Pi.
%\end{equation*}
%We therefore conclude that $(w_p',w_c)\in\B_{nd}$. Furthermore, since $w_p'\in\proj{w_p}{\B_{out}}$ and $\B_{out}\subset\B_{nd}$, we therefore see that $(w_p',w_c)\in\B_{out}$. However, by construction, there also exists $w_p''$ such that $(w_p'',w_c)\in\B_{in}$. This means that $\proj{w_c}{\B_{in}}\cap\proj{w_c}{\B_{out}}\neq\varnothing$. However, we have already proven in \remind{the previous paper} that this is impossible. Therefore, we must have that $\B_{xi}=\varnothing$. This concludes the proof.

\subsection{Proof of Theorem \ref{thm:controlledexist}}\label{appx:proofthm:controlledexist}
\revise{As illustrated by \eqref{eq:wpControlled}, the controlled behavior before the integration of dissipativity conditions is decomposed into a component that we can manipulate (through $z_p$) and a component blocked by the network interconnection that is out of control. The control goal is therefore to further limit the choice of $z_p$ such that
\begin{equation}\label{eq:disscomplete}
    (\F_{pc}z_p+\F_hz_h)^\top\widetilde{\Delta}(\F_{pc}z_p+\F_hz_h)\geq0
\end{equation}
for all $z_h$. This requires firstly the existence of a lower bound with respect to $z_h$, which is equivalent to requiring that
\begin{subequations}\label{eq:minimumSchur}
    \begin{align}
        \F_h^\top\widetilde{\Delta}\F_h&\geq0,\label{eq:minimumSchur1}\\
        (\F_h^\top\widetilde{\Delta}\F_h)_\perp\F_h^\top\widetilde{\Delta}\F_{pc}&=0.\label{eq:minimumSchur2}
    \end{align}
\end{subequations}
Secondly, it requires the guarantee of the existence of $z_p$ such that the lower bound is non-negative, i.e., the existence of $z_p$ such that (by substituting \eqref{eq:minimumSchur2} into \eqref{eq:disscomplete})
\begin{equation}\label{eq:dissWorst}
    z_p^\top\F_{pc}^\top\left[\widetilde{\Delta}-\widetilde{\Delta} \F_h\left(\F_h^\top \widetilde{\Delta} \F_h\right)^\dagger\F_h^\top \widetilde{\Delta}\right]\F_{pc}z_p\geq0.
\end{equation}
If \eqref{eq:minimum} is satisfied, then, by Schur complement, both conditions in \eqref{eq:minimumSchur} are guaranteed and
\begin{equation}
    \F_{pc}^\top\left[\widetilde{\Delta}-\widetilde{\Delta} \F_h\left(\F_h^\top \widetilde{\Delta} \F_h\right)^\dagger\F_h^\top \widetilde{\Delta}\right]\F_{pc}\geq\F_{cc}^\top P^\top\widetilde{\Phi}_cP\F_{cc}.
\end{equation}
In other words, if there exists a $z_p$ such that
\begin{equation}\label{eq:globalcontrolDiss}
    z_p^\top\F_{cc}^\top P^\top\widetilde{\Phi}_cP\F_{cc}z_p\geq0, 
\end{equation}
then \eqref{eq:dissWorst} is satisfied and we have achieved the desired global performance requirement.

Now, for each controller $\Sigma_c^j$ all possible trajectories to choose from are given by \eqref{eq:controllerLocal}. Its local QdF condition can therefore be written as
\begin{equation}\label{eq:localDiss}
    \hat{w}_c^{j\top}\widetilde{\Phi}_c^j\hat{w}_c^j=\begin{bmatrix}
        z_p^j\\ z_m^j
    \end{bmatrix}^\top\begin{bmatrix}
        \F_{cm}^{j\top}\widetilde{\Phi}_c^j\F_{cm}^j & \F_{cm}^{j\top}\widetilde{\Phi}_c^j\F_{cf}^j\\
        \F_{cf}^{j\top}\widetilde{\Phi}_c^j\F_{cm}^j & \F_{cf}^{j\top}\widetilde{\Phi}_c^j\F_{cf}^j
    \end{bmatrix}\begin{bmatrix}
        z_p^j\\ z_m^j
    \end{bmatrix}.
\end{equation}
According to Corollary \ref{cor:dissfree}, the satisfaction of \eqref{eq:DOF} guarantees that $z_p^j$ is free. In other words, for any past trajectory, there exists a local solution $w_c^j(k)$ such that \eqref{eq:localDiss} is non-negative. Furthermore, from the construction of $\F_{cm}$ and those of $\F_{cm}^j$, it is easy to see that $\cs({P\F_{cm}})\subset\cs(\diag\{\F_{cm}^j\}_{j=1}^{N_c})$. As a result, \eqref{eq:DOF} also ensures that
\begin{equation}\label{eq:dissImp}
    \F_{cm}^\top P^\top\widetilde{\Phi}_cP\F_{cm}\geq0.
\end{equation}
Again, from Corollary \ref{cor:dissfree}, this means that $z_p$ is free, i.e., it is always possible to find $z_m$ in \eqref{eq:controllerGlobal} such that $\hat{w}_c^\top P^\top\widetilde{\Phi}_cP\hat{w}_c^\top\geq0$. As noted after Eq. \eqref{eq:controllerGlobal}, each choice of $(z_p,z_m)$ in \eqref{eq:controllerGlobal} has a correspondence of $z_p$ in \eqref{eq:controlImp}. This means that the satisfaction of \eqref{eq:DOF} guarantees the existence of $z_p$ such that \eqref{eq:dissImp} is satisfied for any past trajectories, which leads to the fulfillment of \eqref{eq:dissWorst} by satisfying \eqref{eq:minimum} and, ultimately, the guarantee of the desired global performance condition accommodating the blocked behavior with any past system trajectories. Finally, since $\F_{fk}$ is of full row rank, $w_f(k)$ is free in the uncontrolled behavior. The satisfaction of \eqref{eq:disscomplete} for any $z_h$ means that behavior relating $w_f(k)$ is not restricted in the controlled behavior, hence it is free.}

\bibliographystyle{IEEEtran}

\bibliography{ref,ref2}

\end{document}